\documentclass[aps,pra,10pt,notitlepage,twocolumn,nofootinbib,superscriptaddress]{revtex4-2}

\usepackage[dvipsnames]{xcolor}
\usepackage{framed}
\definecolor{shadecolor}{rgb}{0.9,0.9,0.9}

\usepackage{mathtools}
\usepackage{amsmath}
\usepackage[shortlabels]{enumitem}

\usepackage{graphicx,epic,eepic,epsfig,amsmath,latexsym,amssymb,verbatim,color}
 
\usepackage{amsfonts}       
\usepackage{nicefrac}       

\usepackage{amsmath}
\usepackage{bbm}

\usepackage{float}
\usepackage{tikz}
\usetikzlibrary{chains}
\usetikzlibrary{fit}
\usetikzlibrary{arrows}
\usetikzlibrary{decorations}

\usepackage{epsfig}
\usetikzlibrary{shapes.symbols,patterns} 
\usepackage{pgfplots}

\usepackage[strict]{changepage}
\usepackage{hyperref}
\hypersetup{colorlinks=true,citecolor=blue,linkcolor=blue,filecolor=blue,urlcolor=blue,breaklinks=true}

\usepackage[marginal]{footmisc}
\usepackage{url}
\usepackage{theorem}

\newtheorem{definition}{Definition}
\newtheorem{proposition}{Proposition}
\newtheorem{lemma}{Lemma}

\newtheorem{theorem}[proposition]{Theorem}

\newtheorem{corollary}[proposition]{Corollary}

\def\squareforqed{\hbox{\rlap{$\sqcap$}$\sqcup$}}
\def\qed{\ifmmode\squareforqed\else{\unskip\nobreak\hfil
\penalty50\hskip1em\null\nobreak\hfil\squareforqed
\parfillskip=0pt\finalhyphendemerits=0\endgraf}\fi}
\def\endenv{\ifmmode\;\else{\unskip\nobreak\hfil
\penalty50\hskip1em\null\nobreak\hfil\;
\parfillskip=0pt\finalhyphendemerits=0\endgraf}\fi}
\newenvironment{proof}{\noindent \textbf{{Proof~} }}{\hfill $\blacksquare$}

\newcounter{remark}
\newenvironment{remark}[1][]{\refstepcounter{remark}\par\medskip\noindent%
\textbf{Remark~\theremark #1} }{\medskip}

\newcounter{example}

\mathchardef\ordinarycolon\mathcode`\:
\mathcode`\:=\string"8000
\def\vcentcolon{\mathrel{\mathop\ordinarycolon}}
\begingroup \catcode`\:=\active
  \lowercase{\endgroup
  \let :\vcentcolon
  }

\usepackage{cleveref}
\usepackage{graphicx}
\usepackage{xcolor}

\RequirePackage[framemethod=default]{mdframed}
\newmdenv[skipabove=7pt,
skipbelow=7pt,
backgroundcolor=darkblue!15,
innerleftmargin=5pt,
innerrightmargin=5pt,
innertopmargin=5pt,
leftmargin=0cm,
rightmargin=0cm,
innerbottommargin=5pt,
linewidth=1pt]{tBox}

\newmdenv[skipabove=7pt,
skipbelow=7pt,
backgroundcolor=red!15,
innerleftmargin=5pt,
innerrightmargin=5pt,
innertopmargin=5pt,
leftmargin=0cm,
rightmargin=0cm,
innerbottommargin=5pt,
linewidth=1pt]{rBox}

\newmdenv[skipabove=7pt,
skipbelow=7pt,
backgroundcolor=blue2!25,
innerleftmargin=5pt,
innerrightmargin=5pt,
innertopmargin=5pt,
leftmargin=0cm,
rightmargin=0cm,
innerbottommargin=5pt,
linewidth=1pt]{dBox}
\newmdenv[skipabove=7pt,
skipbelow=7pt,
backgroundcolor=darkkblue!15,
innerleftmargin=5pt,
innerrightmargin=5pt,
innertopmargin=5pt,
leftmargin=0cm,
rightmargin=0cm,
innerbottommargin=5pt,
linewidth=1pt]{sBox}
\definecolor{darkblue}{RGB}{0,76,156}
\definecolor{darkkblue}{RGB}{0,0,153}
\definecolor{blue2}{RGB}{102,178,255}
\definecolor{darkred}{RGB}{195,0,0}

\newcommand{\nc}{\newcommand}
\nc{\rnc}{\renewcommand}
\nc{\lbar}[1]{\overline{#1}}
\nc{\bra}[1]{\langle#1|}
\nc{\cptp}{\operatorname{CPTP}}
\nc{\hptp}{\operatorname{HPTP}}
\nc{\ket}[1]{|#1\rangle}
\nc{\spn}{\operatorname{span}}

\nc{\rep}[1]{\textbf{R}(#1)}

\nc{\ketbra}[2]{|#1\rangle\!\langle#2|}
\nc{\dketbra}[2]{|#1\rangle\!\rangle\!\langle\!\langle#2|}
\nc{\dket}[1]{|#1\rangle\!\rangle}
\nc{\dbra}[1]{\langle\!\langle#1|}
\nc{\dbraket}[2]{\langle\!\langle#1|#2\rangle\!\rangle}

\nc{\braket}[2]{\langle#1|#2\rangle}

\nc{\proj}[1]{| #1\rangle\!\langle #1 |}
\nc{\avg}[1]{\langle#1\rangle}
\nc{\smfrac}[2]{\mbox{$\frac{#1}{#2}$}}
\nc{\tr}{\operatorname{Tr}}
\nc{\ox}{\otimes}
\nc{\dg}{\dagger}
\nc{\dn}{\downarrow}
\nc{\cA}{{\cal A}}
\nc{\cB}{{\cal B}}
\nc{\cC}{{\cal C}}
\nc{\cD}{{\cal D}}
\nc{\cE}{{\cal E}}
\nc{\cF}{{\cal F}}
\nc{\cG}{{\cal G}}
\nc{\cH}{{\cal H}}
\nc{\cI}{{\cal I}}
\nc{\cJ}{{\cal J}}
\nc{\cK}{{\cal K}}
\nc{\cL}{{\cal L}}
\nc{\cM}{{\cal M}}
\nc{\cN}{{\cal N}}
\nc{\cO}{{\cal O}}
\nc{\cP}{{\cal P}}
\nc{\cQ}{{\cal Q}}
\nc{\cR}{{\cal R}}
\nc{\cS}{{\cal S}}
\nc{\cT}{{\cal T}}
\nc{\cU}{{\cal U}}
\nc{\cV}{{\cal V}}
\nc{\cX}{{\cal X}}
\nc{\cY}{{\cal Y}}
\nc{\cZ}{{\cal Z}}
\nc{\cW}{{\cal W}}
\nc{\csupp}{{\operatorname{csupp}}}
\nc{\qsupp}{{\operatorname{qsupp}}}
\nc{\var}{{\operatorname{var}}}
\nc{\rar}{\rightarrow}
\nc{\lrar}{\longrightarrow}
\nc{\polylog}{{\operatorname{polylog}}}
\nc{\wt}{{\operatorname{wt}}}
\nc{\av}[1]{{\left\langle {#1} \right\rangle}}
\nc{\supp}{{\operatorname{supp}}}

\nc{\argmin}{{\operatorname{argmin}}}

\def\x{\xi}

\nc{\RR}{{{\mathbb R}}}
\nc{\CC}{{{\mathbb C}}}
\nc{\FF}{{{\mathbb F}}}
\nc{\NN}{{{\mathbb N}}}
\nc{\ZZ}{{{\mathbb Z}}}
\nc{\PP}{{{\mathbb P}}}
\nc{\QQ}{{{\mathbb Q}}}
\nc{\UU}{{{\mathbb U}}}
\nc{\EE}{{{\mathbb E}}}
\nc{\id}{{\operatorname{id}}}

\nc{\CHSH}{{\operatorname{CHSH}}}

\newcommand{\Op}{\operatorname}
\nc{\be}{\begin{equation}}
\nc{\ee}{{\end{equation}}}
\nc{\bea}{\begin{eqnarray}}
\nc{\eea}{\end{eqnarray}}
\nc{\<}{\langle}
\rnc{\>}{\rangle}
\nc{\rU}{\mbox{U}}

\nc{\ob}[1]{#1}

\nc{\SEP}{{\text{\rm SEP}}}
\nc{\NS}{{\text{\rm NS}}}
\nc{\LOCC}{{\text{\rm LOCC}}}
\nc{\PPT}{{\text{\rm PPT}}}
\nc{\EXT}{{\text{\rm EXT}}}
\nc{\Sym}{{\operatorname{Sym}}}

\nc{\ERLO}{{E_{\text{r,LO}}}}
\nc{\ERLOCC}{{E_{\text{r,LOCC}}}}
\nc{\ERPPT}{{E_{\text{r,PPT}}}}
\nc{\ERLOCCinfty}{{E^{\infty}_{\text{r,LOCC}}}}
\nc{\Aram}{{\operatorname{\sf A}}}

\usepackage{tikz}
\usepackage{hyperref}
\hypersetup{colorlinks=true,citecolor=blue,linkcolor=blue,filecolor=blue,urlcolor=blue,breaklinks=true}

\makeatletter
\def\grd@save@target#1{%
  \def\grd@target{#1}}
\def\grd@save@start#1{%
  \def\grd@start{#1}}
\tikzset{
  grid with coordinates/.style={
    to path={%
      \pgfextra{%
        \edef\grd@@target{(\tikztotarget)}%
        \tikz@scan@one@point\grd@save@target\grd@@target\relax
        \edef\grd@@start{(\tikztostart)}%
        \tikz@scan@one@point\grd@save@start\grd@@start\relax
        \draw[minor help lines,magenta] (\tikztostart) grid (\tikztotarget);
        \draw[major help lines] (\tikztostart) grid (\tikztotarget);
        \grd@start
        \pgfmathsetmacro{\grd@xa}{\the\pgf@x/1cm}
        \pgfmathsetmacro{\grd@ya}{\the\pgf@y/1cm}
        \grd@target
        \pgfmathsetmacro{\grd@xb}{\the\pgf@x/1cm}
        \pgfmathsetmacro{\grd@yb}{\the\pgf@y/1cm}
        \pgfmathsetmacro{\grd@xc}{\grd@xa + \pgfkeysvalueof{/tikz/grid with coordinates/major step}}
        \pgfmathsetmacro{\grd@yc}{\grd@ya + \pgfkeysvalueof{/tikz/grid with coordinates/major step}}
        \foreach \x in {\grd@xa,\grd@xc,...,\grd@xb}
        \node[anchor=north] at (\x,\grd@ya) {\pgfmathprintnumber{\x}};
        \foreach \y in {\grd@ya,\grd@yc,...,\grd@yb}
        \node[anchor=east] at (\grd@xa,\y) {\pgfmathprintnumber{\y}};
      }
    }
  },
  minor help lines/.style={
    help lines,
    step=\pgfkeysvalueof{/tikz/grid with coordinates/minor step}
  },
  major help lines/.style={
    help lines,
    line width=\pgfkeysvalueof{/tikz/grid with coordinates/major line width},
    step=\pgfkeysvalueof{/tikz/grid with coordinates/major step}
  },
  grid with coordinates/.cd,
  minor step/.initial=.2,
  major step/.initial=1,
  major line width/.initial=2pt,
}
\makeatother

\usepackage{thmtools}
\usepackage{thm-restate}
\usepackage{etoolbox}
\makeatletter
\def\problem@s{}
\newcounter{problems@cnt}

\newcommand{\allproblems}{\problem@s}
\makeatother

\usepackage{amsmath,amsfonts,amssymb,array,graphicx,mathtools,multirow,bm,tcolorbox,relsize,booktabs}
\usepackage{duckuments}
\usepackage[shortlabels]{enumitem}
\usepackage{graphicx,epic,times,eepic,epsfig,latexsym,verbatim,color}
\usepackage{nicefrac}       
\usepackage[linesnumbered,ruled,procnumbered]{algorithm2e}
\usepackage[table]{xcolor}
\usepackage{mathrsfs}
\usepackage{framed}
\definecolor{shadecolor}{rgb}{0.9,0.90,0.9}
\usepackage{bbm}
\usepackage{epsfig}
\usepackage{pgfplots}
\usepackage[strict]{changepage}
\usepackage{hyperref}
\hypersetup{colorlinks=true,citecolor=blue,linkcolor=blue,filecolor=blue,urlcolor=blue,breaklinks=true}
\usepackage[marginal]{footmisc}
\usepackage{url}
\usepackage{theorem}
\pgfplotsset{compat=1.18} 

\usepackage{thmtools}
\usepackage{thm-restate}
\usepackage{etoolbox}
\usepackage[utf8]{inputenc}
\usepackage[T1]{fontenc}
\usepackage{quantikz} %

\allowdisplaybreaks
\begin{document}
\title{Programmable Open Quantum Systems}
\author{Mingrui Jing}
\author{Mengbo Guo}
\affiliation{Thrust of Artificial Intelligence, Information Hub,\\
The Hong Kong University of Science and Technology (Guangzhou), Guangzhou 511453, China}
\author{Lin Zhu}
\affiliation{Thrust of Artificial Intelligence, Information Hub,\\
The Hong Kong University of Science and Technology (Guangzhou), Guangzhou 511453, China}
\affiliation{Quantum Science Center of Guangdong-Hong Kong-Macao Greater Bay Area, Shenzhen 518045, China}
\author{Hongshun Yao}
\author{Xin Wang}
\email{felixxinwang@hkust-gz.edu.cn.}
\affiliation{Thrust of Artificial Intelligence, Information Hub,\\
The Hong Kong University of Science and Technology (Guangzhou), Guangzhou 511453, China}

\begin{abstract}
Programmability is a unifying paradigm for enacting families of quantum transformations via fixed processors and program states, with a fundamental role and broad impact in quantum computation and control. While there has been a shift from viewing open systems solely as a source of error to treating them as a computational resource, their programmability remains largely unexplored. In this work, we develop a framework that characterizes and quantifies the programmability of Lindbladian semigroups by combining physically implementable retrieval maps with time‑varying program states. Within this framework, we identify quantum programmable classes enabled by symmetry and stochastic structure, including covariant semigroups and fully dissipative Pauli Lindbladians with finite program dimension. We further provide a necessary condition for physical programmability that rules out coherent generators and typical dissipators generating amplitude damping. For such non‑physically programmable cases, we construct explicit protocols with finite resources. Finally, we introduce an operational programming cost, defined via the number of samples required to program the Lindbladian, and establish its core structural properties, such as continuity and faithfulness. These results provide a notion of programming cost for Lindbladians, bridge programmable channel theory and open‑system dynamics, and yield symmetry‑driven compression schemes and actionable resource estimates for semigroup simulation and control in noisy quantum technologies.
\end{abstract}

\maketitle
\textbf{\textit{Introduction}}.--- Programmability is the ability to encode a desired transformation into a compact program and reliably retrieve it on arbitrary input, and it serves as a unifying principle of computation and control in both classical and quantum settings~\cite{Brown1992field,Nielsen1997programmable}. As quantum technologies move from isolated few-qubit demonstrations to large-scale, noisy devices~\cite{Preskill2018quantum}, an efficient programming paradigm, where fixed processors are supplied with suitable program states to enact families of target maps on data registers, is required~\cite{Nielsen1997programmable,Bisio2010optimal}. 
Unlike the classical case, where arbitrary Boolean functions can be deterministically realized with finite programs, deterministic universal programming of all quantum unitaries with a finite-dimensional program register is impossible~\cite{Nielsen1997programmable,Sedlak2019optimal}. This no-go result has motivated two main directions: probabilistic programming of restricted one-parameter unitary families with optimal success probabilities and quantified overheads~\cite{Vidal2002storing,Sedlak2020probabilistic}, and deterministic schemes that restrict the target family and exploit symmetry to bypass the dimensionality barrier.

A particularly powerful approach in the probabilistic setting is the \textit{port-based} strategy, which enables probabilistic universality with heralded success and program states given by the Choi states of the target channels~\cite{Kaur2017amortized,Bennett1996mixed,Horodecki1999general,Chiribella2009realization}. More recently, symmetry has emerged as a key resource for exact deterministic programmability of infinite families of covariant channels with the optimality being judged through the structure of the commutant and the geometry of the corresponding Choi states~\cite{Gschwendtner2021programmabilityof}. Similar strategies have been applied to extend these results to demonstrate advantages in the storage and retrieval of isometries~\cite{Yoshida2025quantum}. These developments established a quantitative theory of programmability for classes of unitary and non-unitary channels, clarifying the role of symmetry, entanglement, and program dimension as computational resources.

\begin{figure}[h!]
    \centering
    \includegraphics[width=0.85\linewidth]{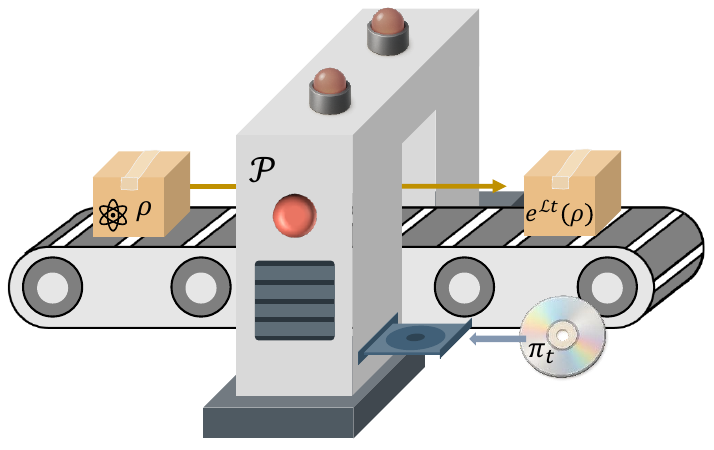}
    \caption{Overview diagram of programming open-system dynamics. By injecting the program $\pi_t$ at time $t$, the processor $\cP$ can reproduce the action of $e^{t\cL}$ onto an arbitrary initial state $\rho_0$.}
    \label{fig:overview_problem}
\end{figure}

However, a large class of physically relevant Markovian quantum dynamics is non-unitary and time-continuous due to the  unavoidably interacting with uncontrolled environments.  Such dynamics is captured by one-parameter semigroups of completely positive and trace-preserving (CPTP) maps governed by the Lindblad master equation~\cite{Lidar2019lecture,Rivas2012open}. Simulating and controlling open-system dynamics is essential for predicting noise-driven behavior in near-term and fault-tolerant devices, optimizing error mitigation and control strategies, and guiding the design, benchmarking, and verification of quantum technologies~\cite{Chen2025,Peng2025quantum,Barreiro2011open,Chen2018simulating,Ding2024simulating,Kamakari2022digital,Weimer2021simulation,Lostaglio2021error}. More recent works have shifted from viewing open systems solely as a source of error (decoherence) to treating them as a computational resource for solving ordinary differential equations~\cite{Shang2025designing} and for preparing quantum Gibbs states~\cite{Chen2025,Guo2025designingopen}. These tasks are inherently interdisciplinary, linking quantum information, computation, and many-body physics, which showcase a vital role of efficient programming of open quantum systems.

In this work, we introduce a framework for characterizing the programmability of the dynamics of open quantum systems and propose an operational notion to quantify the physical programming cost of a Lindbladian. We formulate programmability in terms of injecting a time-varying analytic family of program states into a fixed retrieval map, such that the resulting programmable architecture produces the entire Lindbladian semigroup for all times (see Fig.~\ref{fig:overview_problem}). We identify a class of quantum programmable systems by examining the intrinsic symmetry and stochastic features of their Lindbladians with finite program dimension. Beyond these constructive cases, we establish a necessary structural condition for physical programmability that rules out broad classes of dynamics, including coherent generators and paradigmatic dissipative channels such as amplitude damping. For such non-programmable generators, we derive explicit quasi-quantum sampling protocols that achieve exact programming, and examine their performance via numerical sampling.

Building on the notion of physical implementability for simulating maps beyond quantum operations~\cite{Jiang2021physical}, we propose an advanced notion of programming cost within the framework of semidefinite programming (SDP). We prove several  properties of this cost, including faithfulness, unitarily invariance, monotonicity, and an analytic continuation principle that extends exact programmability from any nontrivial time interval to all times. We also investigate the tradeoff between the cost and the error threshold through numerical experiments. Our results provide the first quantitative bridge between programmable channel theory and open system dynamics, elevating symmetry and invariant subspaces to compression mechanisms for entire dynamical semigroups. This framework offers practically relevant guidelines for resource-efficient simulation and control of open quantum systems, and suggests new cross-disciplinary strategies for engineering programmable quantum devices.


\vspace{2mm}
\textbf{\textit{Programmability of Lindbladian}}.---
Consider $\cL$ a physical Lindbladian superoperator acting on the principal system $\cD(\cH_S)$. Let $\rho_0\in \cD(\cH_S)$ be an arbitrary initial state. The dynamics of the system at any time $t\geq 0$ are fully determined by the Lindblad master equation: $\frac{d\rho}{dt} = \cL(\rho)$. The solution to the equation can be expressed as $\rho(t) =e^{t\cL}(\rho_0)$, where $e^{t \cL}$ forms a quantum dynamical semigroup generated by $\cL$~\cite{Rivas2012open}, and for any $t\geq 0$, $e^{t \cL}$ is a quantum channel. 

We target to construct a  programmable processor that solves the master equation with in a time interval, provided arbitrary initial states. Let $\cH_{tot} = \cH_S \ox \cH_P$ be a composite quantum system. A programmable quantum processor consists of a linear map $\cP$ from $\cH_{tot}$ to $\cH_{S'}$ with $\cH_{S'} \simeq \cH_{S}$, and a continuous set of program states $\pi_t$ acting on $\cH_P$ of dimension $d_P$. Denote $\Omega$ as the allowable subset of linear maps. We say $\cL$ is $\Omega_{\epsilon}$-programmable within a time interval $[0,T]$ for some $T\geq 0$, if there exists one-parameter, continuous set of program states $\pi_t \in \cD(\cH_P)$ and a $\cP\in\Omega$ such that $\forall t \in [0, T]$, 
\begin{equation*}
    \frac{1}{2}\|\cP(\cdot \ox \pi_t) - e^{t \cL} \|_{\diamond} \leq \epsilon.
\end{equation*}
We call $\{\cP, \pi_t\}_t$ a $\Omega_{\epsilon}$-programming protocol of $\cL$ within $[0,T]$, and $\cP$ the programming channel, (or retrieval channel). If such a protocol exists as $T\rightarrow \infty$, we say $\cL$ is $\Omega_{\epsilon}$-programmable. When $\epsilon$ vanishes, we call it the (exact) $\Omega$-programmable scenario. In particular, if a CPTP channel $\cP$ can be constructed, we say $\cL$ is \textit{quantum programmable}. If such a construction only exists for Hermitian-preserving and trace-preserving (HPTP) map $\cP$. We say $\cL$ is \textit{quasi-quantum programmable}.
The channel programming problem has been investigated in several previous literature~\cite{Yoshida2025quantum,Gschwendtner2021programmabilityof,Nielsen1997programmable} for which our definition of programmability is closely related to the $\mathbb{S}$-set retrieving problem by treating the entire semigroup generated by $\cL$ as a subset of quantum channels. 

Given $\cA_t=e^{t \cL}$ as the quantum dynamical semigroup generated by $\cL$, a clear observation showcases that if there exists a finite set of fixed, time-independent quantum channels $\{\cE_j\}_{j=1}^K$ such that, for any time $t\geq 0$, $\cA_t$ can be decomposed into a convex combination $\cA_t = \sum_j p_j(t) \cE_j$ where $p_j(t)$'s are the time-dependent probability amplitudes. Then, based on the measure-and-prepare strategy~\cite{Gschwendtner2021programmabilityof}, we can construct $\pi_t = \sum_j p_j(t) \ketbra{j}{j}$ within a $K$-level qudit system. By measuring the system of program states with respect to the computational basis, the outcome can tell which channel $\cE_j$ is going to be applied. The entire configuration, therefore, forms a CPTP programming protocol. 

From the perspective of generator space, if the composition of the Lindbladian on a fixed set of quantum channels can be represented as a reversed transition-rate matrix (Q-matrix)~\cite{Norris1998markov}, provided $\cI \in \Op{conv}[\{\cE_j\}_j]$. Then, the Lindbladian equation reduces to the classical master equation which leads to the solution $\cA_t = \sum_j p_j(t) \cE_j$ and $p_j(t)$ is evolved under the stochastic matrix generated by the Q-matrix. This fact can force a class of Lindbladians to be quantum programmable.
\begin{theorem}
    Given $\cL$ be any $n$-qubit fully dissipative Lindbladian where all the jump operators are Pauli operators. Then, $\cL$ is quantum programmable.
\end{theorem}
Pauli-Lindbladians serve as canonical testbeds for studying mixing, contractivity, and controlled systems, and they underpin efficient noise identification and benchmarking by reducing parameter estimation to a few axis-aligned decay constants~\cite{Van2024techniques,Van2023probabilistic,Schwartzman2025modeling,Chruifmmode2016generalized}. This theorem provides a convenient basis for modeling and simulating general Pauli dynamics in quantum devices by storing the evolution information into the static quantum resources. The details and proofs have be developed in Appendix~\ref{appendix:isolated_system}. 

Meanwhile, symmetry plays a central role across quantum information theory, quantum resource theory, and channel simulation~\cite{Hayashi2006quantum,Chitambar2019quantum}. The group structures enables the compression of complex dynamics into symmetry-invariant subspaces, facilitating reuse, reducing algorithmic overhead for simulation and estimation, which yields deeper physical insight and more elegant mathematical formulations. Let $G$ be a compact symmetry group and $U$ be a unitary representation on Hilbert spaces $\cH_S$. Then, the semigroup $\cA_t$ is said to be \textit{covariant}~\cite{Holevo1995structure,Chruscinski2021universal} if for any $A\in\cB(\cH_S)$ and any element $g\in G$, the equality, $\cA_t(U_g A U_g^{\dagger}) = U_g \cA_t(A) U_g^{\dagger}$
for any time $t\geq 0$. This is equivalent to say the Lindbladian $\cL$ is commute with the unitary evolution of $U_g$ for all $g$. 

\begin{corollary}\label{coro:cptp_programmable_covariant}
  Let $\cL$ be a covariant Lindbladian. Then, $\cL$ is quantum programmable with the $d_P$ determined by the commutant of $\overline{U}\ox U$, where $U$ is an irrep of a compact group.
\end{corollary}
The proof directly follows by theorem 25 of~\cite{Gschwendtner2021programmabilityof}. Briefly, the covariant condition implies a block diagonal form to the commutant space of $\overline{U} \ox U$, which forces the dynamics generated by $\cL$ lying in an invariant subspace (polytope)~\cite{Song2022positive,Yorke1967invariance} within all the quantum channels.

From corollary~\ref{coro:cptp_programmable_covariant}, we have identified that for covariant Lindbladian, the corresponding dynamics can be quantum programmable where the dimension of program states is restricted. A particular example is the isotropic depolarizing system where its dynamics can be represented as a time-dependent depolarizing channel being covariant  under the action of $\Op{SU}(d)$. In fact, symmetry provides a unifying framework to classify and constrain the dynamics of open quantum systems, shaping their spectra, steady states, and universal behaviors~\cite{Kawabata2023symmetry,Siudzinska2022phase,Snamina2020dynamical}. It also enables robust control and protected responses in noiseless subsystems  for Lindbladian dynamics~\cite{Albert2016geometry}. 

\vspace{2mm}
\textbf{\textit{Programming dynamics beyond deterministic quantum operations}}.--- The CPTP-programmability provides a distinct convenience in simulating Lindbladian dynamics. However, not every open systems can be fully storaged and retrieved using a parameterized static resource. In this section, we discuss some typical Lindbladians that can never be programmed via a single quantum channel. Let $\cL$ be a quantum-programmable Lindbladian acting on $\cD(\cH_d)$. Necessarily, there exists a quantum channel $\cE$, and a constant $\alpha \geq 0$ such that for any  $\rho\in\cD(\cH_d)$, $\cL(\rho) = \alpha( \cE(\rho) - \rho)$.
This form provides a direct condition to check whether a given $\cL$ is quantum programmable. In particular, we can identify that the following dynamics can not be programmed via a  quantum channel.
\begin{proposition}\label{prop:non-cptp_coherent}
    Any $\cL$ contains non-trivial coherent part, i.e., $H\not\propto I_d$, is not quantum programmable.
\end{proposition}
The inherent Hamiltonian evolution would introduce the imaginary spectrum of $\cL$ which can not be simply programmed via a physical channel. This naturally extends the previous results on the no-go theorem of storing and retrieving the quantum phase gate~\cite{Sedlak2020probabilistic,Vidal2002storing} to the scenarios of the one-parameter group evolutions generated by any Hamiltonian involving the interaction with thermal bath.

A second example that is not quantum programmable is the Lindbladian of the spontaneous emission system.
\begin{proposition}\label{prop:non-cptp_ad}
    Let $\cL$ be defined by a single jump operator $L = \sqrt{\gamma}\ketbra{0}{1}$ where $\gamma$ is the emission rate. Then, $\cL$ is not quantum programmable.
\end{proposition}
This model is used to capture the photonic loss which is essential for understanding how dissipation and decoherence shape light–matter interactions, enabling accurate modeling of open quantum systems and non-Hermitian dynamics. The corresponding semigroup is the famous amplitude damping channel in quantum information theory. 
The detailed proof of the necessary condition can be found in Appendix~\ref{appendix:program_beyond_physical_channel}, and the further proofs for proposition~\ref{prop:non-cptp_coherent} and~\ref{prop:non-cptp_ad} can be found in Appendix~\ref{appendix:isolated_system} and~\ref{appendix:proofs_of_photonic_loss}, respectively. 

\begin{figure}[!ht]
    \centering
    \includegraphics[width=0.9\linewidth]{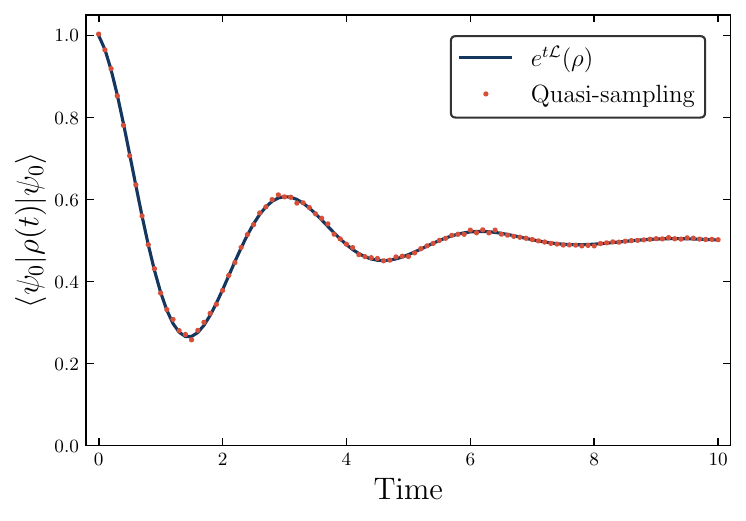}
    
    \caption{Examination of the accuracy of the quasi-sampling method for programming the SWAP–dephasing Lindbladian with $\lambda=0.5$. The solid line represents the exact analytical result, while the scattering points denote quasi-sampled estimates for each time point. The simulation is performed over the time interval $T\in[0,10]$, demonstrating excellent agreement between quasi-sampling and the exact  dynamics.}
    \label{fig:time_evolution_Z}
\end{figure}

Desipte the non-CPTP programmability, many quantum computational tasks ultimately end with measurements, and the goal is to estimate expectation values rather than access the full output state. This motivates the development of the quasi-quantum methods, in which the action of a non-physical map $\cE$ can be `virtually' simulated via Monte-Carlo sampling of its quasi-decomposition of physical CPTP channels. As a result, the expectation value of $\tr(O\cE(\rho))$ can be estimated by the post-processing and reweighting~\cite{Jiang2021physical}. We take advantage of this perspective by considering HPTP maps as our programming primitives that can be implemented through quasi-quantum technique.
\begin{proposition}\label{prop:hptp_protocol_arbitrary_H}
Let $\cL = i\Op{ad}_H$ for some $d$-dimensional Hermitian operator $H$ with $K$ distinct eigenvalues $(\lambda_1, \cdots, \lambda_K)$, where $\Op{ad}_H = [H,\cdot]$. Then, there exists an quasi-sampling protocol with the corresponding program states $\pi_t$ of dimension $K$ to exactly program  $\cL$.
\end{proposition}
Proposition~\ref{prop:hptp_protocol_arbitrary_H} extends the results from the qubit gate example in~\cite{Vidal2002storing}. It shows that the dimension of the program states for programming a purely coherent $\cL$ is determined by the number of degenerate spaces of $H$. With the known spectral information of $H$, a quasi-quantum protocol can be constructed for any purely coherent dynamics. 

For the photonic loss system, we have taken the ideas from the circuit knitting technique~\cite{Piveteau2023circuit}, for which the program states consists of $6$ parameterized components. 
\begin{proposition}\label{hptp_protocol_ad}
Let $\cL$ be defined by a single jump operator $L = \sqrt{\gamma}\ketbra{0}{1}$ where $\gamma$ is the emission rate. Then, there exists an exact quasi-sampling protocol with the corresponding program states $\pi_t$ of dimension $12$.
\end{proposition}
The entire processor requires a quasi-sampling and post-selection on the measurement outcomes from the ancillary system. Above all, the HPTP maps demonstrate a strong advantage in programming general quantum dynamics compared with physical channels. All the detailed derivations of proposition~\ref{prop:hptp_protocol_arbitrary_H} and~\ref{hptp_protocol_ad} and protocol circuits have been collected in Appendix~\ref{appendix:isolated_system} and~\ref{appendix:proofs_of_photonic_loss}, respectively. 

A more interesting example of a two-qubit open-system dynamics combines coherent exchange interaction with dissipative collective dephasing. The competition between SWAP-generated exchange interaction and Bell-basis dephasing provides a minimal setting for studying information scrambling and chaotic relaxation in noisy quantum systems~\cite{Zanardi2021}. Its Lindbladian reads~\cite{Zanardi2021}, $\mathcal{L} = i\Op{ad}_S + \lambda \left( \mathcal{D}_{\mathbb{B}} - \cI \right)$,
where $S$ is the SWAP operator, $\mathcal{D}_{\mathbb{B}}$ is the completely dephasing channel
in the Bell basis $\mathbb{B}$, and $\lambda$ is the dephasing strength. Since $[\Op{ad}_S, \cD_{\mathbb{B}}\,]=0$, the resulting quantum dynamical can be expressed in, $\mathcal{E}_t 
= e^{-\lambda t} e^{it\Op{ad}_S}
+ \bigl(1-e^{-\lambda t}\bigr) \cD_{\mathbb{B}}$,
where $e^{it\Op{ad}_S}(\rho) = e^{itS}\rho e^{-itS}$. By adding ancillary system, a quasi-quantum programming protocol can be constructed for the system combined with proposition~\ref{prop:hptp_protocol_arbitrary_H}.

We further perform numerical simulations to examine the performance of the quasi-quantum protocol for reproducing the target time evolutions of the two-qubit chaotic Lindbladian. Setting the initial state $\ket{\psi_0}=\ket{01}$, the overlap between the initial state and the evolved state can be theoretically computed $\bra{\psi_0}\rho(t)\ket{\psi_0}= \frac{1}{2}(1 + e^{-\lambda t}\cos(2t))$. As shown in figure~\ref{fig:time_evolution_Z}, our programming protocol demonstrates an excellent agreement with the exact values, indicating the potential of our framework as a computational resource to simulate and control open-systems.

\vspace{2mm}
\textbf{\textit{Quantifying programming cost of Lindbladian}.}--- Quantify the cost of computational resources of programming quantum dynamics is important in both operational and practical senses. In the previous literature, the cost, or the resources of programming quantum channels is reasonably determined by the minimum dimension of the program states~\cite{Gschwendtner2021programmabilityof,Yoshida2025quantum} required for achieving the programming tasks. Though, it is usually hard to estimate and demand case-by-case analysis. We take the inspiration from the physical implementation of
HPTP maps using quasi-technique~\cite{Jiang2021physical}. In particular, the cost of realizing quasi-probability decomposition of any linear map is quantified by the sampling overhead, defined as the $1$-norm of the decomposition coefficients. This cost has broad applications in circuit knitting, quantum error mitigation and many-body dynamical simulations~\cite{Piveteau2022quasiprobability,Zhao2024retrieving,Gherardini2024quasiprobabilities}. Meanwhile, the diamond norm of any TP maps has been proved to be equal to their implementation cost~\cite{Regula2021operational}.

We define the \textit{programming cost} of $\cL$ as follows. Let $\epsilon\geq 0$, the $\epsilon$-error \textit{programming cost} of the superoperator $\cL$ with respect to the program states $\pi_t$ for $t\in[0,T]$ is defined as
\begin{equation}\label{eq:primal_program_pga}
\begin{aligned}
    \gamma_{\epsilon}(\pi_t,\cL, T)&:=\log \min_{\cP \in \Omega}\Big\{||\cP||_\diamond\Big\arrowvert \\&\forall t\in[0,T], \frac{1}{2}\left\|\cP(\cdot \otimes\pi_t) - e^{t \cL}\right\|_{\diamond} \leq \epsilon\Big\}.
\end{aligned}
\end{equation}
Here, we take the binary logarithm in the definition. Given the construction of $\pi_t$, the cost~\eqref{eq:primal_program_pga} can be estimated via solving convex optimization by sampling a sufficient number of time steps. Notice that, there is no guarantee for the cost to be strictly bounded, and different choices of $\pi_t$ can significantly vary the cost values. We say $\pi_t$ is \textit{veritable} if $\gamma_{\epsilon}(\pi_t, \cL,T) < \infty$, and $\cL$ is $\Omega_{\epsilon}$-programmable within $[0,T]$. One can show that the cost $\gamma_{\epsilon}(\pi_t, \cL,T)$ is monotone in $T$. Therefore, if the programming cost can be bounded above by taking $T\rightarrow \infty$, the entire dynamics generated by $\cL$ can be $\Omega_{\epsilon}$-programmable, and we can switch the order of limit and minimization to define, $\gamma_{\epsilon}(\pi_t,\cL) :=\lim_{T\rightarrow \infty}\gamma_{\epsilon}(\pi_t,\cL, T)$.  

\begin{figure}[!ht]
    \centering
    \includegraphics[width=.45\textwidth]{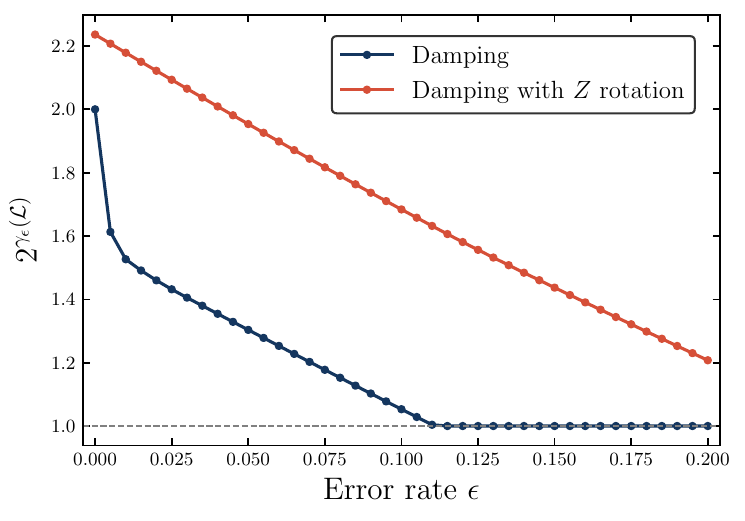}

    \caption{Numerical estimations on the port-based $\epsilon$-error programming cost $\gamma_{\epsilon}$ of two Lindbladians with respect to the $\epsilon\in[0,0.2]$. The blue and red curves correspond to programming the damping Lindbladian and the damping Lindbladian supplemented with a $Z$ rotation, respectively. The time interval is fixed to $[0,10]$.}
    \label{fig:error_threshold}
\end{figure}

In general, searching for the veritable set of program states can be challenging which requires an infinite precision of solving convex optimization programming~\cite{Banchi2020convex}. To ensure the physical feasibility of preparing the program states, instead of encountering the extreme theoretical conditions, we make the assumption that the states $\pi_t$ must be \textit{analytically} defined, i.e., the map $\pi:[0,\infty)\rightarrow \cD(\cH_P)$ is \textit{(real) analytic}~\cite{Krantz2002primer,Ahlfors1979complex}. Since the exponential map $t\rightarrow e^{tM}$ is analytic. The duality of linear maps defined on the Hilbert space ensures the map $t\rightarrow e^{t \cL}$ is also analytic and due to the Choi isomorphism. 

An important choice of analytic program states is to use Choi states $J(e^{t \cL})$ at time $t$, which is a well-known configuration in the channel simulation problem, and is closely related to the concept of \textit{teleportation-simulatability} in the quantum information theory~\cite{Kaur2017amortized,Bennett1996mixed,Horodecki1999general,Chiribella2009realization}. With the importance of Choi states in the studies of quantum channels and dynamical resources~\cite{Gour2019quantify,Regula2021one,Hsieh2021communication}, we, particularly, write $\gamma_{\epsilon}(\cL)$ in the expression of programming cost if the $\pi_t = J(e^{t \cL})$. We also omit writing $\epsilon$ as a subscript when taking $\epsilon=0$, which corresponds to the exact programming scenario. The numerical results for investigating the effects on the cost values with respect to the error threshold $\epsilon$ have been illustrated in figure~\ref{fig:error_threshold}. With the port-based program states, we examine the cost variation from a purely photonic loss system and the one with an intrinsic $Z$-rotation. One can observe the latter system is much harder to program than the purely damping system.

Given any Lindbladian $\cL$ acting on $\cD(\cH_d)$, the programming cost has the following properties:
    \begin{enumerate}
        \item (\textbf{Analytic continuation}) Let $\pi_t$ be the analytic program states defined for $t\geq 0$. Then, $\gamma(\pi_t, \cL) = \gamma(\pi_t, \cL, T)$ for any $T>0$.
        \item (\textbf{Faithfulness}) $\gamma(\pi_t, \cL) \geq 0$ and $\gamma(\pi_t, \cL) = 0$ iff. $\cL$ is CPTP-programmable.
        \item For any invertible quantum noise $\cN$ acting on $\cD(\cH_P)$, $\gamma(\pi_t, \cL) \leq \gamma(\cN(\pi_t), \cL) \leq \gamma(\pi_t, \cL) + \log\|\cN^{-1}\|_{\diamond}$.
    \end{enumerate}
The faithfulness of programming cost delivers that when $\gamma(\pi_t, \cL) = 0$, there exists a  quantum channel and the corresponding program states $\pi_t$ to simulate the dynamics generated by $\cL$ for all time. In the practical scenario, when the program states are unavoidably 
influenced by the invertible noise, the programming cost can increase. Interestingly, any closed evolution $\cU$ on the $\cH_P$ can keep the programming cost invariant since $\|\cU^{-1}\|_{\diamond} = 1$. The detailed proofs on the properties of the programming cost can be found in Appendix~\ref{appendix:properties_cost}.


\textbf{\textit{Concluding remarks}}.--- We have introduced a general framework for programming open quantum system dynamics by combining physically implementable retrieval maps with analytic, time-varying program states. Our theoretical analysis led to an operational notion and properties of programming cost. We further identified symmetry and stochastic structure as central resources for quantum programmability of open quantum systems, while a necessary condition for quantum programmability rules out broad classes. In terms of applications, our explicit protocols with finite, quantified sampling cost indicate that programmable open-system architectures can become practical tools for semigroup simulation and control. This motivates further analytical refinements and concrete circuit constructions tailored to noisy quantum platforms.

One may already notice that $\mathcal{L}$ is quantum programmable if the corresponding invariant subspace of the dynamical equation $\frac{d\mathcal{A}_t}{dt} = \mathcal{L}\circ \mathcal{A}_t$ with $\mathcal{A}_0 = \mathcal{I}$ can be fully covered by a fixed polytope in the space of quantum channels. Determining nontrivial invariant subspaces of any dynamical system is \textit{undecidable}~\cite{Yorke1967invariance}. Further research seeks an efficient algorithm to identify the invariant subspaces of the dynamics, or even for a restricted set of input states. On the other hand, the optimality of these HPTP programming protocols in terms of the programming cost is not settled. Apart from the Choi-state resource, an efficient construction of other analytic program states is an interesting direction that needs further exploration.

\textbf{\textit{Acknowledgments.}}---
This work was partially supported by the National Key R\&D Program of China (Grant No.~2024YFB4504004), the National Natural Science Foundation of China (Grant. No.~12447107), the Guangdong Provincial Quantum Science Strategic Initiative (Grant No.~GDZX2403008, GDZX2403001), the Guangdong Provincial Key Lab of Integrated Communication, Sensing and Computation for Ubiquitous Internet of Things (Grant No. 2023B1212010007).


\newpage
\bibliography{ref}

\numberwithin{equation}{section}
\renewcommand{\theequation}{S\arabic{equation}}
\renewcommand{\theproposition}{S\arabic{proposition}}
\renewcommand{\thedefinition}{S\arabic{definition}}
\renewcommand{\thefigure}{S\arabic{figure}}

\newpage
\vspace{2cm}
\onecolumngrid
\vspace{2cm}

\begin{center}
\large{\textbf{Supplemental Material for
`Programmable Open Quantum Systems'}}
\end{center}

In this Supplemental Material, we offer detailed proofs of the theorems and propositions in the manuscript `Programmable Open Quantum Systems'. In Section~\ref{appendix:Notations and preliminaries}, we will introduce the basics of master equation, quantum dynamical semigroups and the Liouville representation~\ref{appendix:lindblad_to_choi}. We then, review the concepts of physical implementability of general linear maps in~\ref{appendix:physical_implementability}. In Section~\ref{appendix:programmability of oqs}, we generally discuss the programmability of open quantum systems by demonstrating the properties of the programming cost in~\ref{appendix:properties_cost}. We particular discuss the programmability of Lindbaldian dynamics with the Choi state resources in~\ref{appendix:choi_state_resource}. In Section~\ref{sec:program_invariant_subspaces}, we discuss the several cases of Lindbladians that can be physically programmed with a vanishing programming cost. At last, in Section~\ref{appendix:program_beyond_physical_channel}, we discuss the cases stated in the main contents, including the Lindbladians that is not CPTP-programmable. Furthermore, we prove the HPTP-programmability of these dynamics by constructing the feasible protocols. In~\ref{appendix:isolated_system}, we discuss the details of programming Lindbladians of isolated systems. In~\ref{appendix:proofs_of_programming_dephasing_exchange}, we discuss the details about a Lindbaldian with both coherent and dissipative parts, serving as the model examined in the numerical experiment. In~\ref{appendix:proofs_of_photonic_loss}, we discuss an HPTP-programmable Lindbladian corresponding to the photonic loss system, and is purely dissipative.

\section{Notations and preliminaries}\label{appendix:Notations and preliminaries}

Let $\cH_d$ be a $d$-dimensional Hilbert space. Throughout the entire manuscript, we label the total programming system with the  associated Hilbert spaces $\cH_S$ and $\cH_P$, i.e., the principal and the program state spaces respectively, with the corresponding dimensions $d_S$ and $d_P$. We use $\cB(\cH_d)$ to represent the set of all bounded linear operators acting on system $S$, and $\cP(\cH_d)$ to represent the set of Hermitian and positive semidefinite operators. For any linear operator $A\in \cB(\cH_d)$, we denote $\overline{A}$ as the complex conjugation of $A$, $A^T$ as the tranpose of $A$. The operation $A^{\dagger}$ is, therefore, defined as taking both complex conjugation and transpose to the original operator $A$, denoted as $A^{\dagger}$. A linear operator $\rho \in \cP(\cH_d)$ is considered a density operator of a quantum state if its trace is equal to one. We denote the set of all density operators on system $S$ as $\cD(\cH_d)$. 
A pure quantum states $\psi$ is a rank-one density operator. In this manuscript, we use ket-bra notation to denote the state vector $\ket{\psi} \in \CC^d$ where $\psi = \ketbra{\psi}{\psi}$.
The maximally entangled state of dimension $d\otimes d$ is denoted as $\Phi_d=1/d\sum_{i,j=0}^{d-1} \ketbra{ii}{jj}$. The trace norm of $\rho$ is denoted as $\|\rho\|_1 = \tr(\sqrt{\rho^\dagger \rho})$. We write a general linear map that transforms any operators from $\cB(\cH_A)$ to $\cB(\cH_B)$ as $\cN$, (or $\cN_{A\rightarrow B}$), and denote the set of all linear maps from $\cB(\cH_A)$ to $\cB(\cH_A)$ as $\mathscr{L}(A\rightarrow B)$. In particular, the set of all \textit{Hermitian-preserving and trace-preserving} (HPTP) linear maps, denoted as $\hptp(A\rightarrow B)$. The set of all quantum channels, denoted as $\cptp(A\rightarrow B)$, contains the linear transformations from $\cB(\cH_A)$ to $\cB(\cH_B)$ that is \textit{completely positive and trace-preserving} (CPTP). Clearly, $\cptp(A\rightarrow B) \subseteq \hptp(A\rightarrow B) \subseteq \mathscr{L}(A\rightarrow B)$. 

The Choi isomorphism of the linear map $\cN$ is represented as $J(\cN) \in \cB(\cH_A\ox \cH_B)$, which is defined as $d(\cI_{A} \ox \cN)(\Phi_{AA'})$ where $\cI$ is the identity map and $\Phi_{AA'}$ is the maximally entangled state acting on $\cH_{AA'}$. When dealing with Choi operators, composition of linear maps and their actions on linear operators can be conveniently expressed
in terms of the \emph{link product}~\cite{Chiribella2008quantum}, which will be denoted as $\star$. For maps $\cE\in\mathscr{L}(A\rightarrow B)$ and $\cF\in\mathscr{L}(B\rightarrow C)$.
The link product $J(\cE) \star J(\cF)$ gives the Choi operator of the composite map $\cF \circ \cE$ is defined as
\begin{equation*}
   J(\cE) \star J(\cF) := \tr_B \Big[ \big(J(E)^{T_B}\otimes I_C\big)\big(I_A\otimes J(\cF)\big)\Big], 
\end{equation*}
with $T_B$ the partial transposition on $\mathcal{H}_B$ and $\tr_B$ the partial trace on $\mathcal{H}_B$.  The link product is commutative, 
$J(\cE) \star J(\cF) = J(\cF) \star J(\cE)$, and associative, $J(\cA)\star(J(\cB)\star J(\cC)) = (J(\cA)\star J(\cB))\star J(\cC)$. 

\subsection{Lindbladian master equation, quantum dynamical semigroup and its  representations}\label{appendix:lindblad_to_choi}
Open quantum systems are described by the evolution of density operators under completely positive and trace-preserving dynamics that account for irreversible exchange of information and energy with an environment. Under assumptions of weak coupling, short bath correlation times, and time-homogeneous driving, the dynamics becomes Markovian in the sense that the future depends only on the present state. In this regime the time evolution is generated by the  Gorini–Kossakowski–Sudarshan–Lindblad (GKSL) master equation~\cite{Lidar2019lecture,Rivas2012open}, also called the Lindblad master equation. If $\rho(t)$ denotes the system density operator, the evolution is given by,
\begin{equation}\label{eq:lindblad_master_equation}
    \frac{d}{dt}\rho = -i[H,\rho] + \sum_j \gamma_j\left(L_j \rho L_j^{\dagger} - \frac{1}{2}\left\{L_j^{\dagger} L_j, \rho\right\}\right) = \cL(\rho),
\end{equation}
where $H$ is an effective Hamiltonian, which may include environment-induced Lamb shifts, and $L_j$ are the jump operators that encode the dissipative channels. The superoperator $\cL\in\mathscr{L}(A\rightarrow A')$ is referred to as the Lindbladian. The first term in Eq.~\eqref{eq:lindblad_master_equation} is named the coherent part and the remaining term is called the dissipative part. In a formal setup, one can write the commutator into the adjoint representation $\Op{ad}_H(\rho) = [H,\rho]$ with $H$ being a Hermitian operator. The associated family of linear maps $\cA_t = e^{t \cL}$ defines the continuous-time evolution through $\rho(t) = \cA_t(\rho(0))$ for $t \geq 0$. The family of maps $\cA_t$ forms a quantum dynamical semigroup in the sense that $\cA_0$ is the identity map.
\begin{definition}[quantum dynamical semigroup]
    Let $\cD(\cH)$ be the density operator space acting on a Hilbert space $\cH$. A quantum dynamical semigroup is a collection of bounded (super)operators $\cA_t$ for $t\geq 0$ satisfying:
    \begin{enumerate}
        \item $\cA_0(\rho) = \rho \; \forall \rho \in \cD(\cH)$,
        \item $\cA_{t+s}(\rho) = \cA_t \circ 
        \cA_s(\rho)$ for all $\rho \in \cD(\cH)$,
        \item $\cA_t$ is completely positive for all $t\geq 0$,
        \item $\cA_t$ is $\sigma$-weakly continuous operator for all $t\geq 0$,
        \item For all $\rho \in \cD(\cH)$, the map $t\mapsto \cA_t(\rho)$ is continuous with respect to the $\sigma$-weak topology.
        \end{enumerate}
\end{definition}

The structure of the Lindbladian guarantees complete positivity and trace preservation of the evolution $\cA_t$ at all times, and each quantum dynamical semigroup is uniquely generated (Hille-Yosida theorem). In finite dimensions, any strongly continuous one-parameter semigroup of completely positive and trace-preserving maps has a generator of GKSL form, and conversely any such generator yields a well-defined evolution~\cite{gyamfi2020fundamentals,Song2022positive,Agredo2022kossakowski}. 

For analysis and computation it is often advantageous to pass to Liouville space, also called operator space, by vectorizing operators. In this representation one maps an operator $A$ to a vector $\dket{A}$ in a space of dimension $d^2$ equipped with the Hilbert–Schmidt inner product $\dbraket{A}{B} = \tr(A^{\dagger}B)$. Linear maps on operators become matrices acting on these vectors, so that the commutator with the Hamiltonian becomes a simple difference of Kronecker products and the dissipative terms factorize into tensor products of jump operators and their adjoints. To be clear, in this manuscript, we use the convention of vectorization as $\dket{\ketbra{i}{j}} = \ket{i}\ket{j}$. By using the property $\dket{ABC} = (A\ox C^T)\dket{B}$, the operator $-i[H,\cdot]$ can be computed as,
\begin{equation*}
    -i\dket{[H, \rho]} = -i\dket{H\rho - \rho H} = -i(H\ox I - I\ox H^T) \dket{\rho}.
\end{equation*}
Similarly for the dissipation part, we have,
\begin{equation*}
    \dket{L_k \rho L_k - \frac{1}{2}\{L_k^{\dagger} L_k, \rho\}} = (L_k \ox \overline{L}_k - \frac{1}{2}(I\ox L_k^{\dagger} L_k) - \frac{1}{2}((L_k^{\dagger} L_k)^T\ox I) \dket{\rho}.
\end{equation*}
We now define the superoperator $\bm{L}$,
\begin{equation}
    \bm{L} = -i(H \ox I - I \ox H^T) + \sum_k \left(L_k \ox \overline{L}_k - \frac{1}{2}(I\ox L_k^{\dagger} L_k) - \frac{1}{2}((L_k^{\dagger} L_k)^T\ox I)\right).
\end{equation}
By vectorizing the density operator $\rho \rightarrow \dket{\rho}$, this representation turns the master equation into a linear ordinary differential equation in a larger vector space, enabling spectral analysis, Krylov subspace propagation, and steady-state solvers.
\begin{equation}
    \frac{d}{dt} \dket{\rho} = \bm{L}\dket{\rho} \Rightarrow \dket{\rho(t)} = e^{\bm{L} t} \dket{\rho(0)}.
\end{equation}
The generator $\bm{L}$ has a spectrum confined to the closed left half-plane in the complex numbers, so that the real parts of its eigenvalues are nonpositive. The $\Op{ker}[\bm{L}]$ consists of stationary states, and the negative of the largest nonzero real part defines a spectral gap that controls asymptotic mixing rates when no persistent oscillations survive. Eigenvalues with purely imaginary parts correspond to reversible components of the dynamics supported on decoherence-free subspaces or noiseless subsystems; when the peripheral spectrum reduces to the eigenvalue zero, the evolution converges exponentially to a unique stationary state that is full rank, a property often called primitivity in the context of irreducible dynamics. Conditions for uniqueness and mixing, developed by Evans and Frigerio and by Fagnola and Rebolledo among others, relate the noise algebra generated by the jump operators to the existence of a faithful steady state~\cite{Frigerio1977quantum,Fagnola2001existence,Yoshida2024uniqueness,Baumgartner2012structures}.

Notice that the $d^2 \ox d^2$ matrix $\bm{K}(t) = e^{\bm{L} t}$ is the corresponding channel representation of $\cA_t$ at time $t$. Such a matrix is similar to the Choi matrix of the channel up to a reshuffling of basis. Let $\cE_t$ be the channel at time $t$. Recalling its Choi matrix,
\begin{equation*}
    J(\cE_t) = \sum_{i,j=0}^{d-1} \ketbra{i}{j} \ox \cE_t(\ketbra{i}{j}).
\end{equation*}
Each element $\cE_t(\ketbra{i}{j})$ can be derived from $\bm{K}(t)\dket{\ketbra{i}{j}}$. Denote $K_{ab,ij}(t)$ as the components of $\bm{K}(t)$ where $(a,b)$ is the row indices. Subsititude into the Choi matrix to derive,
\begin{equation}
    J(\cE_t) = \sum_{i,j,a,b = 0}^{d-1} K_{ab,ij}(t) \ketbra{i}{j} \ox \ketbra{a}{b}.
\end{equation}

\subsection{Physical implementability of general linear maps}\label{appendix:physical_implementability}
Physically realizable quantum operations are exactly CPTP. Many linear maps of interest in quantum information, such as positive but not completely positive maps, or inverses of noisy channels, are not CPTP and thus are not directly implementable. In previous literature, Jiang et. al.~\cite{Jiang2021physical} has devised a general technique to simulate the actions of linear maps by accessing quasi-sampling strategy. The key idea is to express an HPTP map $\cN$ as a signed linear combination of CPTP maps. 
\begin{equation*}
    \cN = \sum_{\alpha} \eta_{\alpha} \cE_{\alpha} \; \text{with} \; \cE_{\alpha} \in \cptp\; \alpha \in \RR.
\end{equation*}
Equivalently, the optimal decomposition can always be taken with two channels. The non-physicality, or the simulation cost is defined by the minimal total weight of this quasiprobability:
\begin{equation*}
    \nu(\cN):=\log\min\{\eta_+ +\eta_- \;\big|\; \cN = \eta_+ \cE_{+} - \eta_- \cE_{-} \; \text{with} \; \cE_{\pm} \in \cptp\; \eta_{\pm} \geq 0\}.
\end{equation*}
A relaxation replacing CPTP by CPTN (completely positive, trace non-increasing) yields the same optimum. This quantity satisfies the following properties:
\begin{itemize}
    \item Faithfulness: $\nu(\cN) = 0$ iff. $\cN$ is CPTP.
    \item Additivity under tensor product, $\nu(\cN\ox \cM) = \nu(\cN) + \nu(\cM)$. 
    \item Subadditivity under channel composition: $\nu(\cN\circ\cM) = \nu(\cN) + \nu(\cM)$.
    \item Unitarily invariance: $\nu(\cU \circ \cN \circ \cV) = \nu(\cN)$ for $\cU,\cV$ two unitary operations.
\end{itemize}
This quantity has operational meaning in probabilistic error mitigation. For any invertible noise channel $\cN$ with inverse $\cN^{-1}$ (generally not CPTP), one uses the optimal quasiprobability decomposition of $\cN^{-1}$ to mitigate the noise effect from estimating expectation values. The sampling overhead is propositional to  $2^{\nu(\cN^{-1})}$. Consequently, $\nu(\cN^{-1})$ gives the fundamental lower bound on sampling cost. Furthermore, from the investigation of the operational application of the diamond norm~\cite{Regula2021operational}, when the object $\cN$ is TP or, more generally proportional to a TP map, it holds that, $\nu(\cN) = \log\|\cN\|_{\diamond}$
where $\|\cdot\|_{\diamond}$ is the diamond norm defined for general linear maps.

\section{Programming Lindbladian with intrinsc symmetry patterns}\label{sec:program_invariant_subspaces}

From the differential equation Eq.~\eqref{eq:de_master_eq_linear_maps}, provided a fixed Lindbladian $\cL$, there is a unique solution $\cA_t = e^{t \cL}$ that solves the equation with initial condition $\cA_0 = \cI$. Using the measure-and-prepare strategy, one can observe that if the entire dynamical semigroup is closed within a polytope in the quantum channel space. This can be extended to the problem of \textit{invariant subspace} of the differential equation~\eqref{eq:de_master_eq_linear_maps}. However, finding the invariant subspaces of a dynamical system is generally difficult (NP-hard), especially for nonlinear systems or complex linear operators~\cite{Blondel2000boundedness,Blondel2000survey}. In the following discussion, we will focus on specific cases of interests to delve the problem in the framework of physical programmability.

In the space of quantum channels, invariant subspaces arise whenever a family of maps is closed under convex mixing and composition, so that the action of admissible processing cannot leave the family, for example,  a convex hull $C = \Op{conv}{\cE_k}{k=1}^m$ of fixed channels, where any element $\cE \in C$ has the form,
\begin{equation*}
    \cE(\rho) = \sum_{k=1}^m p_k \cE_k(\rho)
\end{equation*}
for $\sum_k p_k = 1$ and $p_k \geq 0$, and $C$ is invariant under convex combinations and, when $\{\cE_k\}$ is closed under composition, under semigroup generation as well. Such convexly generated sets define linear cones in the Liouville space that are stable under coarse-graining and represent operationally meaningful subspaces.
\begin{definition}[$Q$-matrix~\cite{Norris1998markov}]
    Let $\mathbb{I}$ be a countable set. A $Q$-matrix on $\mathbb{I}$ is a matrix $Q=(q_{ij}:i,j\in \mathbb{I})$ satisfying the following conditions:
    \begin{enumerate}
        \item $0\leq -q_{ii} < \infty$ for all $i$, where $q_{ii} = \sum_{i\neq j} q_{ij}$;
        \item $q_{ij} \geq 0$ for all $i\neq j$;
        \item $\sum_{i\in \mathbb{I}} q_{ij} = 0$ for all $i$.
    \end{enumerate}
\end{definition}

\begin{lemma}\label{lem:decomposible_generator}
    Let $\cS = \{\cE_1, \cE_2, \cdots, \cE_k\}$ be a finite set of linearly independent quantum channels on $\cB(\cH_d)$. Then, the dynamic semigroup $\{\Phi_t\}_{t\geq 0} \subseteq \Op{conv}[\cS]$ if the following conditions holds, 
    \begin{enumerate}
        \item $\cI\in \Op{conv}[\cS]$.
        \item There exists a $k\times k$ $Q$-matrix $Q$, such that $\cL\circ\cE_i = \sum_{j=1}^k q_{ji} \cE_j$.
    \end{enumerate}
\end{lemma}
\begin{proof}
    Let $\cS$ be such a set satisfying two conditions. Denote $c_j(t)$ to be the time-dependent real coefficients such that $c_j(0) = \lambda_j$ for $j=1,\cdots, k$. We propose the solution $\Phi_t$ of form: $\Phi_t = \sum_{i=1}^{k} c_i(t) \cE_i$. 
    It suffices to show $\bm{c}(t) = (c_1(t), \dots, c_k(t))^T$ must be a probability vector for all $t \ge 0$. Substituting our proposed form into the differential equation, 
    \begin{equation*}
    \begin{aligned}
        \frac{d}{dt}\Phi_t = \frac{d}{dt} \sum_{j=1}^{k} c_j(t) \cE_j &= \sum_{j=1}^{k} \frac{dc_j(t)}{dt} \cE_j =  \sum_{j=1}^{k} c_j(t) \cL \circ \cE_j = \sum_{m=1}^k \left(\sum_{j=1}^k q_{mj}c_j(t)\right)\cE_m \\
        &\Rightarrow \frac{dc_m(t)}{dt} = \sum_{j=1}^k q_{mj} c_j(t),
    \end{aligned}
    \end{equation*}
    which reduces to  the (classical) master equation:
    $\frac{d\bm{c}(t)}{dt} = Q \bm{c}(t)$, and the formal solution to this differential equation is reads
    $\bm{c}(t) = e^{Q t} \bm{c}(0)$. The matrix $P(t) = e^{Q t}$ are stochastic matrices for all $t \ge 0$, with each column summing to $1$~\cite{Norris1998markov}. Now, 
    $\bm{c}(0)$ is a probability vector, their product $\bm{c}(t) = P(t)\bm{c}(0)$ is also a probability vector for all $t \ge 0$. Because $\bm{c}(t)$ remains a probability vector for all $t \ge 0$,  $\Phi_t = \sum_i c_i(t) \cE_i$ is always a convex combination of the vertices $
    \cE_i$, which completes the proof.
\end{proof}

\begin{theorem}
    Given $\cL$ be any $n$-qubit fully dissipative Lindbladian with all the jump operators defined by Pauli's. Then $\cL$ is $\cptp$-Programmable.
\end{theorem}
\begin{proof}
    Let $\cL$ be a fully dissipative Lindbladian with $L_j = P_j$ where $P_j$'s are $n$-qubit Pauli operators, and $\gamma_j \geq 0$ be the dissipation rates. Denotes $\cP_k(\rho) = P_k \rho P_k$. Then, for the off-diagonal terms in the Pauli basis, 
    \begin{equation*}
    \begin{aligned}
        \cL\circ \cP_k(\rho) &= \sum_{j=1}^{4^n-1} \gamma_j (P_j P_k \rho P_k P_j - \frac{1}{2}\{P^{\dagger}_jP_j, P_k\rho P_k\})\\ = &\sum_{j=1}^{4^n-1} \gamma_j (P_j P_k \rho P_k P_j - P_k\rho P_k)= \sum_{l\neq k} \gamma_{j} P_{l} \rho P_{l} - \Gamma P_k \rho P_k,
    \end{aligned}
    \end{equation*}
    where $\Gamma = \sum_j \gamma_j$, $P_l = c_{jk} P_jP_k$ with $|c_{jk}| = 1$ and we denote $l = j \sim k$ as a re-labeling of $4^n$ Pauli operators including the identity. We can now define $q_{kk} := -\Gamma; \quad q_{jk} := \gamma_{j\sim k}$ for $j\neq k$. It suffices to verify the conditions for the $Q$-matrix. But, $-q_{kk} = \Gamma \geq 0$, $\sum_{j\neq k} q_{jk} = \sum_j \gamma_j = -\Gamma = q_{kk}$, and,
    \begin{equation*}
        \sum_{j=0}^{4^n - 1} q_{jk} = q_{kk} + \sum_{j\neq k} q_{jk} = -\Gamma + \sum_{j\neq k} \gamma_{j\sim k} = -\Gamma + \Gamma = 0,
    \end{equation*}
    where the map $j \mapsto j \sim k$ for a fixed $k$ is a permutation. As $j$ runs over all indices except $k$, the index $j = j \sim k$ runs over all indices except $0$. Therefore, $q_{jk}$ forms a $4^n\times 4^n$ $Q$-matrix and by Lemma~\ref{lem:decomposible_generator}, we have completed the proof.
\end{proof}

Within this general picture, covariant channels provide a symmetry-defined invariant subspace. We start with the most general definition of covariance.
\begin{definition}[$UV$-covariant quantum channel~\cite{Gschwendtner2021programmabilityof}] Let $G$ be a compact group and let $U$ and $V$ be representations on Hilbert spaces $\cH_1$ and $\cH_2$. Let $\cN:\cB(\cH_1) \rightarrow \cB(\cH_2)$ be a linear map. We call $\cN$ $UV$-covariant if
\begin{equation*}
    \cN(U_g A U_g^{\dagger}) = V_g\cN(A) V_g^{\dagger} \quad \forall A \in \cB(\cH_1), \; \forall g\in G.
\end{equation*}    
\end{definition}

\begin{lemma}
    Let $\cL$ be a Lindbladian and $\cA_t:=e^{t \cL}$ be the generated dynamic semigroup. Then for any fixed unitary evolution $\cU(\cdot) = U(\cdot) U^{\dagger}$, $\cU^{\dagger} \circ \cA_t \circ \cU$ is a valid dynamic semigroups.
\end{lemma}
\begin{proof}
    Let $\cU$ be some fixed unitary evolution and define $\Phi_t = \cU^{\dagger} \circ \cA_t \circ \cU$. Recalling the definition of quantum dynamical semigroup, conditions 3,4,5 are automatically satisfied. Besides, $\Phi_0 = \cU^{\dagger} \circ \cI \circ \cU = \cI$; For any input state $\rho$, the semigroup property,
    \begin{equation*}
        \Phi_{t+s}(\rho) = \Phi_{t} \circ \Phi_{s}(\rho) = \cU^{\dagger} \circ \cA_{t} \circ \cU \circ \cU^{\dagger} \circ \cA_{s} \circ \cU(\rho) =  \cU^{\dagger} \circ \cA_{t+s} \circ \cU(\rho),
    \end{equation*}
    as required.
\end{proof}

\begin{lemma}\label{lem:lindbladian_covariant}
    Let $G$ be a compact group, $\cL$ be a Lindbladian and $\cA_t$ be the dynamical semigroup generated by $\cL$. Then $\cA_t$ is $UU$-covariant for all $t\geq 0$ iff. $\cL$ is $UU$-covariant.
\end{lemma}
\begin{proof}
    For the ($\Rightarrow$) direction, we have,
    \begin{equation*}
        \cA_t(U_g X U_g^{\dagger}) = U_g \cA_t(X) U_g^{\dagger}\quad \forall X \in \cB(\cH_1)\; \forall g\in G\; t\geq 0.
    \end{equation*}
    Taking the derivative at $t = 0$ of both sides, since $U_g$ are time-independent, we derive,
    \begin{equation*}
        \cL(U_gXU_g^{\dagger})=\frac{d}{dt}\Big|_{t=0}\cA_t(U_g X U_g^{\dagger}) = U_g \frac{d}{dt}\Big|_{t=0}\cA_t(X) U_g^{\dagger} = U_g\cL(X)U_g^{\dagger},
    \end{equation*}
    which showcases that $\cL$ is $UU$-covariant.

    For the ($\Leftarrow$) direction, we define another dynamical semigroup $\Phi_t(X) = U_g^{\dagger}e^{t \cL}(U_g X U_g^{\dagger}) U_g$  for some $g\in G$. To show $\Phi_t = \cA_t$, it suffices to show they have the same infinitesimal generators. But
    \begin{equation*}
        \frac{d}{dt}\Big|_{t=0} \Phi_t(X) = U_g^{\dagger}\cL(U_gXU_g^{\dagger})U_g = \cL(X) = \frac{d}{dt}\Big|_{t=0} \cA_t(X),
    \end{equation*}
    for any $X\in\cB(\cH_1)$. Therefore, we end the proof.
\end{proof}

The $UU$-covariant property of dynamical semigroup is fully characterized by its corresponding generator.
\begin{definition}[Commutant]
    Let $\cM$ be a matrix algebra on the Hilbert space $\cH_d$. Its commutant is 
    \begin{equation*}
        \cM' := \{B | BA = AB \; \forall A\in \cM \; \forall t \geq 0\}
    \end{equation*}
\end{definition}
Let $U$ be a unitary representation of a compact group $G$ on $\cH_1$, which can be written as $\cH_1 = \bigoplus_{k=1}^K (\cH_k \ox \cH_k')$ such that $U_g = \bigoplus_{k=1}^K U_g^{(k)} \ox I_{n_k}$ for all $g\in G$ where $U^{(k)}$, $k\in\{1, \cdots, K\}$, are \textit{irreps} of $G$. Furthermore, the corresponding operator algebra generated by $\{U_g\}_{g\in G}$ and its commutant can be written into the direct sum decomposition~\cite{Simon1996representations}.

\begin{lemma}[Lem.11~\cite{Gschwendtner2021programmabilityof}]\label{lem:covariant_channel_choi}
    The covariance property of a channel $\cT\in\cT_{UV}$ w.r.t. the unitary representations $U,V$ of a group $G$ is equivalent to  $[J(\cT), \overline{U}_g \ox V_g] = 0, \forall g\in G$. Here, $\cT_{UV}$ denotes all $UV$-covariant channels.
\end{lemma}

We find that the covariance of a quantum dynamic group $\{\cA_t\}_{t\geq 0}$ is mainly dominated by the infinitesimal generator $\cL$. In the following discussion, we denote
\begin{equation}
    \cJ_{\cA}:= \{J(\cA_t) \in \cB(\cH_1\ox \cH_2) : J(\cA_t) := d(\cI \ox \cA_t)(\ketbra{\Phi}{\Phi})\}
\end{equation}

\begin{lemma}\label{lem:choi_covariant_qds}
    Let $G$ be a compact group and $U$ be a  representation on $\cH$, respectively. If $\cL$ is $UU$-covariant, and $\cA_t = e^{t \cL}$, then $[J(\cA_t), \overline{U}_g \ox U_g] = 0$ for all $t\geq 0$.
\end{lemma}
\begin{proof}
    The proof directly follows by Lemma~\ref{lem:lindbladian_covariant} and Lemma~\ref{lem:covariant_channel_choi}. Since $\cL$ is $UU$-covariant, $\cA_t$ is $UU$-covariant for any $t\geq0$, and therefore, the whole set $\cJ_{\cA}$ does.
\end{proof}

We study representations of the form $U\ox V$ with $U_g\in \cU_1$, the unitary representation group on $\cH_1$, $g\in G$. We denote the commutant
\begin{equation*}
\begin{aligned}
    \cK:=\cM(\overline{U}\ox V)' &= \{X\in \cB(\cH_1\ox \cH_2)\ | \ [X, \overline{U}_g \ox V_g] = 0 \ \forall g\in G\}\\
    &= \bigoplus_{k=1}^K I_{b_k} \ox \cB(\cH_k'),
\end{aligned}
\end{equation*}
where $b_k$ labels each irrep in the direct sum of $\cK$.

\begin{lemma}
    Let $U$ be an irrep of a compact group $G$ on $\cH$. If $\cL$ is $UU$-covariant. Then, 
    the corresponding dynamic semigroup $\cA_t$ is unital for any $t\geq0$, i.e., $\cL(I) = 0$.
\end{lemma}
\begin{proof}
    Let $\cL$ to be $UU$-covariant. From Lemma~\ref{lem:choi_covariant_qds}, $[J(\cA_t), \overline{U}_g\ox U_g] = 0$ for all $g\in G$ and $t\geq 0$. We get,
    \begin{equation*}
        \tr_1(J(\cA_t)) = \tr_1((\overline{U}_g\ox U_g) J(\cA_t)(\overline{U}_g\ox U_g)^{\dagger}) = U_g \tr_1(J(\cA_t)) U_g^{\dagger}
    \end{equation*}
    for any $g\in G$, and therefore, $[U_g, \tr_1(J(\cA_t))] = 0$. We assume $U$ is irrep. According to the Schur's Lemma, $\tr_1(J(\cA_t)) = \lambda I_{d_2}$ for some $\lambda\in\CC$ for any time $t\geq 0$. Taking the trace on the above equation to have,
    \begin{equation*}
        d_1 = \tr(\tr_1(J(\cA_t))) = \lambda d_2 \Rightarrow \lambda = 1,
    \end{equation*}
    as $\cH_1 \simeq \cH_2$. Hence, we find $\cA_t$ is unital for any $t\geq 0$. Now,  taking the time derivative on the equation $\cA_t(I) = I$ at $t=0$ to complete the proof.
\end{proof}

In fact, from the studies of open quantum systems, the annihilation of the identity operator from $\cL$ is sufficient and necessary for $\cA_t$ to be unital~\cite{Rivas2012open}.

\begin{lemma}
    Let $U$ be an irrep of a compact group $G$. If $\cL$ is $UU$-covariant. Then, there exists a finite number of time-dependent probability amplitudes $p_j(t) \geq 0$, and $\sum_{j} p_j(t) = 1$ for $t\geq 0$ such that the corresponding $\cA_t = \sum_{j} p_j(t) \cE_j$, where $\{\cE_j\}_j$ are some fixed, time-independent $UU$-covariant channels.
\end{lemma}
\begin{proof}
    Let $\cL$ be $UU$-covariant so that by Lemma~\ref{lem:lindbladian_covariant}, $\{\cA_t\}_{t\geq 0}$ lies as subset of the space of all $UU$-covariant channels from $\cH_1$ to $\cH_2$. According to Lemma 16~\cite{Gschwendtner2021programmabilityof}, such a set of channels is isomorphic to a polytope $P\subset \RR^K$ for some integer $K$. Taking all $J(\cE_j) \in \Op{ext}(P)$ being the set of extreme points of $P$, then, there exist $p_j(t)$ for any $t\geq 0$, such that $\sum_{j=1}^K p_j(t) = 1$, and the Choi operator of $J(\cA_t)$ can be represented as,
    \begin{equation*}
        J(\cA_{t}) = \sum_j p_j(t) J(\cE_j)
    \end{equation*}
    due to the definition of convex polytope.
\end{proof}

The covariant quantum dynamical semigroup with generator $\cL$ commuting with the symmetry action, so the entire trajectory remains in the symmetry-invariant subspace. This enforces a simultaneous block structure, yielding symmetry-resolved steady states, decoherence-free subspaces, noiseless subsystems, and selection rules for decay and oscillatory modes. In practice, covariance of Lindbladian provides an efficient strategy for simulating the generated dynamics.

\section{Programming Lindbladian dynamics beyond physical channels}\label{appendix:program_beyond_physical_channel}

Apart from CPTP-programmble systems, in the most cases, we can not find a CPTP programming channel to simulate. In this section, we will discuss the conditions that the CPTP-programmable Lindbladian $\cL$ must hold. We will first prove the main proposition~\ref{prop:sufficient_cond_cptp} to derive a necessary condition for $\cL$ to be CPTP-programmable. After that, we will discuss the typical systems that can not be physically programmed. Moreover, we will showcase the HPTP-programmability of these systems by constructing  feasible protocols in this framework.

\begin{proposition}[necessary condition]\label{prop:sufficient_cond_cptp}
    Let $\cL$ be a Lindbladian acting on $\cD(\cH_d)$. Then, $\cL$ is quantum programmable if there exists a quantum channel $\cE$, and a fixed $\alpha \geq 0$ such that for any state $\rho\in\cD(\cH_d)$, $\cL(\rho) = \alpha( \cE(\rho) - \rho)$.
\end{proposition}  

\begin{proof}
Suppose $\cL$ is $\cptp$-programmable, and let $\{\cP, \pi_t\}_t$ be the programming protocol, so that,
\begin{equation*}
    \cP(\rho \ox \pi_t) = e^{t \cL}(\rho),\; \forall \rho\in \cD(\cH_d)\; \forall t \geq 0.
\end{equation*}
Taking the gradient with respect to $t$ at $t=0$, we can derive:
\begin{equation*}
    \cL(\rho)=\frac{d}{dt}\Big|_{t=0}\cP(\rho\ox\pi_t) = \cP(\rho\ox \frac{d\pi_t}{dt}\Big|_{t=0}) = \cP(\rho\ox A),
\end{equation*}
where we denote $A = \frac{d\pi_t}{dt}|_{t=0}$. Since $\pi_t\in\cD(\cH_P)$, and the conservation of probability guarantees that $A$ is both Hermitian and traceless. Suppose $A=\sum_k a_k \ketbra{k}{k}$, denoted the eigen-decomposition where $a_k\in\RR$'s are the eigenvalues. Now write $\cP$ into the Kraus representation $\{E_j\}_j$ for $E_j:\cH_S\ox \cH_P\rightarrow \cH_S$ the Kraus operators. Expand $E_j = \sum_{k} L_{jk} \ox \bra{k}$ with the eigenbasis of $A$. 
The trace-preserving condition of $\cP$ gives,
\begin{equation*}
\begin{aligned}
    \sum_j E_j^{\dagger} E_j = I_S \ox I_P =\sum_{kl} \left(\sum_jL_{jk}^{\dagger}L_{jl}\right) \ox \ketbra{k}{l} \Rightarrow \sum_{j} L_{jk}^{\dagger} L_{jl} = \delta_{kl}I_S.
\end{aligned}
\end{equation*}
We can also derive the action of $\cP$ using the expansion to obtain,
\begin{equation*}
\begin{aligned}
    \cL(\rho) &= \sum_j E_j (\rho \ox A) E_j^{\dagger} = \sum_{j}\sum_{m} a_m \sum_k L_{jk} \ox \bra{k} (\rho \ox \ketbra{m}{m}) \sum_l L_{jl}^{\dagger} \ox \ket{l} \\
    &= \sum_{m}\sum_{jkl} a_m L_{jk} \rho L_{jl}^{\dagger} \delta_{km} \delta_{ml} = \sum_k a_k \sum_j L_{jk} \rho L_{jk}^{\dagger}.
\end{aligned}
\end{equation*}
Notice that for any fixed $k$, from the previous condition, we have $\sum_{j} L_{jk}^{\dagger} L_{jk} = I_S$, satisfying the trace-preserving condition, therefore, defines a valid quantum channel acting on $\cD(\cH_S)$. Denote each of them $\cE_k$ with the Kraus operators $\{L_{jk}\}_j$, we have,
\begin{equation*}
    \cL(\rho) = \sum_k a_k \cE_k(\rho) = \sum_{k, a_k \geq 0} a_k \cE_k(\rho) + \sum_{k, a_k < 0} a_k \cE_k(\rho).
\end{equation*}
Denote $\alpha = \sum_{k, a_k\geq 0} a_k; \beta = \sum_{k, a_k < 0} |a_k|$. Clearly, $\alpha=\beta$. We define $\cE = \frac{1}{\alpha}\sum_{k, a_k \geq 0} a_k \cE_k$, and $\cF = \frac{1}{\alpha}\sum_{k, a_k < 0} |a_k| \cE_k$, to be the positive and negative component quantum channels, respectively, so that, $\cL = \alpha(\cE - \cF)$.

We equate the derived form $\cL = \alpha(\cE - \cF)$ with the canonical Lindblad form:
\begin{equation*}
    \cL(\rho) = \Theta(\rho) - \frac{1}{2}\{\Theta^{\dagger}(I_d), \rho\} - i[H,\rho],
\end{equation*}
where $\Theta$ is some completely positive (CP) map, representing the 'jump' processes; $\Theta^{\dagger}$ is the adjoint map of $\Theta$, and $H$ is some Hermitian operator. Since $\alpha \cE$ is a CP map, we can identify the jump superoperator $\Theta = \alpha \cE$. This implies the remaining terms must match:
\begin{equation*}
    -\alpha \cF(\rho) = - \frac{1}{2}\{\Theta^\dagger(I), \rho\} - i[H, \rho].
\end{equation*}
Using the fact that the dual map of a CPTP map is CP unital. We then, have,
\begin{equation*}
\begin{aligned}
    \alpha \cE(\rho) - \alpha \cF(\rho) &= \alpha \cE(\rho) - \alpha \rho - i[H,\rho]\\
    \Rightarrow \cF(\rho) &= \rho + \frac{i}{\alpha}[H,\rho].
\end{aligned}
\end{equation*}

Notice that, $\cF$ is CPTP. The map $i[H,\cdot]$ is CP iff. $H\propto I_d$ by Lemma~\ref{lem:commutator_cp}. We remain to show that the RHS. is CPTP iff. $H\propto I_d$. The backward direction automatically holds as,
\begin{equation*}
    \cF(\rho) = \rho + \frac{ic}{\alpha}[I_d, \rho] = \cI(\rho),
\end{equation*}
where $H = cI_d$. 

For the forward direction, suppose $H\neq cI_d$. The Choi state of $\cF$ is,
\begin{equation*}
    J({\cF})/d = \Phi_1 - i[I_d\ox H, \Phi_1].
\end{equation*}
Let $\{\ket{\Phi_j}\}_{j=1}^{d^2}$ form an orthonormal basis where $\ket{\Phi_1}$ is the normalized maximally entangled state. We now express $J({\cF})$ into this basis by denoting $a_{kl} = \bra{\Phi_k} J({\cF}) \ket{\Phi_l}$. Notice that
\begin{equation*}
    a_{11} = \bra{\Phi_1} (\Phi_1 - i[I_d\ox H, \Phi_1]) \ket{\Phi_1} = 1.
\end{equation*}
For $k=1,l\geq 2$, the orthogonality implies,
\begin{equation*}
    a_{1l} = \bra{\Phi_1} (\Phi_1 - i[I_d\ox H, \Phi_1]) \ket{\Phi_l} = i\bra{\Phi_1}(I_d\ox H)\ket{\Phi_l}.
\end{equation*}
By Hermiticity, $a_{l1} = a_{1l}^*$. All remaining terms of $k,l\geq 2$ vanish due to the orthogonality, and we derive,
\begin{equation*}
    J({\cF})/d = \begin{pmatrix}
        1 & a_{12} & a_{13} & \hdots\\\
        a_{12}^* & 0 & 0 & \hdots\\
        a_{13}^* & 0 & 0 & \hdots\\
        \vdots & \vdots & \vdots & \ddots
    \end{pmatrix}.
\end{equation*}
Now, we aim to solve for the eigenvalues $\lambda$ by assuming $J({\cF})\ket{\phi} = \lambda \ket{\phi}$ where $\ket{\phi} = c_1\ket{\Phi_1} + \sum_{k>1} c_k\ket{\Phi_k}$. We get a system of equations:
\begin{equation*}
\begin{cases}
    \text{For}\; k=1:\; & c_1 + \sum_{l=2} a_{1l} c_l = \lambda c_1\\
    \text{For}\; k> 1:\; & a_{k1}c_1 = \lambda c_k
\end{cases}.
\end{equation*}
Assuming $\lambda \neq 0$, we have $c_k = a_{k1}c_1 / \lambda$. Substitude this into the first equation to derive,
\begin{equation*}
    c_1\left(1+\frac{1}{\lambda} \sum_{l=2} |a_{1l}|^2 - \lambda\right) = 0 \implies c_1 = 0\; \text{or}\; 1+\frac{1}{\lambda} \sum_{l=2} |a_{1l}|^2 - \lambda = 0.
\end{equation*}
We can assume $c_1 \neq 0$ as the eigenvectors spanned by the orthogonal subspace of $\ket{\Phi_1}$ also spans the kernel of $J({\cF})$. Let $S = \sum_{l=2} |a_{1l}|^2$. We can then solve for the eigenvalues as,
\begin{equation*}
    \lambda_{\pm} = \frac{1 \pm \sqrt{1+4S}}{2}.
\end{equation*}
Since $S \geq 0$, and $S = 0$ iff. all the off-diagonal terms $a_{1l} = 0$, and this only happens when $(I\ox H)\ket{\Phi_1}$ has no component in the orthogonal subspace of $\ket{\Phi_1}$, or $H\propto I_d$. As we assume $H\neq cI_d$. We have $S>0$ and  $\lambda_- < 0$, and $\cF$ is not CP.

Above all, we have shown that $\cF = \cI$ and we conclude
\begin{equation*}
    \cL(\rho) = \alpha(\cE(\rho) - \rho),
\end{equation*}
as required.
\end{proof}

\subsection{Proofs of programming dynamics of isolated systems}\label{appendix:isolated_system}

A particular class of Lindbladians is those without a dissipative part. Let $\gamma_j = 0$ for all $j$, the system becomes isolated and the dynamics are governed by the unitary evolution $\cU_t(\rho_0) = e^{-iHt}(\rho_0) e^{iHt}$. 
\begin{lemma}\label{lem:commutator_cp}
    Let $H$ be a $d$-by-$d$ Hermitian operator and denote its adjoint map as $\cE(\cdot) = -i[H,\cdot]$. Then, $\cE$ is CP iff. $H\propto I_d$.
\end{lemma}
\begin{proof}
    The backward condition is trivially satisfied as $\cE$ becomes the zero map. For the forward direction, derive the Choi state,
    \begin{equation*}
        J({\cE})/d = \frac{1}{d}\sum_{ij} \ketbra{i}{j} \ox \cE(\ketbra{i}{j}) = -\frac{i}{d}\sum_{ij} \ketbra{i}{j} \ox [H, \ketbra{i}{j}] = -i[I_d\ox H, \Phi], 
    \end{equation*}
    where $\Phi = \ketbra{\Phi}{\Phi}$ and $\ket{\Phi} = \frac{1}{\sqrt{d}}\sum_j \ket{jj}$. Since $\cE$ is Hermitian-preserving, $J({\cE})$ is Hermitian and have real spectrum $(\lambda_j)_j$. Notice that $\tr(J({\cE})) = 0$ by using the property of commutator, which implies $\sum_j \lambda_j = 0$. Assuming $\cE$ is CP can cause a contradiction, as there must exist $\lambda_j <0$. 
\end{proof}

This Lemma showcases that for any $H$ that is not proportional to the identity operator, there exists no program states and a CPTP programming channel to retrieve the entire dynamics. However, we can still estimate the cost in the scenario. By  defining suitable program states $\pi_t$, one can estimate the value of the cost by solving the following optimization problem,
\begin{equation}
\begin{aligned}
&\underline{\textbf{Primal Program}}\\
    2^{\gamma(\pi_t, \cL)}=\min\;&p_1+p_2\\
    {\rm s.t.}\;\;& J^{\cP}:=J_1-J_2,\\
    &\tr_{E}[J^{\cP}(\pi_t^{T} \ox I_{SS'})] = d\dketbra{U_t}{U_t}_{SS'},\, \forall t \in \RR,\\
    &J_1\geq 0,\,\tr_{S'}[J_1] =p_1 I_{ES},\\
    &J_2\geq 0,\,\tr_{S'}[J_2] =p_2 I_{ES},
\end{aligned}
\end{equation}
where $U_t = e^{iHt}$. We start by considering the program state $\pi_t = \dketbra{U_t}{U_t}/d$, i.e., the Choi state of the evolution. From the previous study, we have derived that there exists an optimal solution to the problem by raising the symmetry,
\begin{equation}
    [J^{\cP}_{SPS'}, U_t \ox \overline{U}_t\ox U_{\tau} \ox \overline{U}_{\tau}] = 0\quad\forall t,\tau\in\RR.
\end{equation}
Assuming the Hamiltonian $H = \sum_{\alpha} E_{\alpha} \ketbra{E_\alpha}{E_\alpha}$ where $\{\ket{E_\alpha}\}$ form an orthonormal eigenbasis. Then, $e^{iHt} = \sum_{\alpha}e^{iE_{\alpha}t} \ketbra{E_\alpha}{E_\alpha}$. The Choi state reads,
\begin{equation}
\begin{aligned}
    (I\ox U_t) \dketbra{I}{I} (I\ox U_{-t}) = \frac{1}{d}\sum_{\alpha,\beta} \ketbra{E_\alpha}{E_\beta} \ox e^{i(E_{\alpha} - E_{\beta})t}\ketbra{E_\alpha}{E_\beta}.
\end{aligned}
\end{equation}
We then have to investigate the representation theory of the one-parameter group $e^{iHt}$. 
\begin{lemma}
    For a Hermitian operator $H\in\cB(\cH_d)$, and a linear operator $A\in\cB(\cH_d)$, $[A,H] = 0$ if and only if $[e^{iHt},A] = 0$ for all $t\in\RR$.
\end{lemma}
\begin{proof}
    For the ($\Rightarrow$) part, suppose $[A,H] = 0$. Then $[A, H^j] = 0$ for any power $j\in\NN$. Expanding the matrix exponentiation to get,
    \begin{equation*}
        [e^{iHt}, A] = \sum_{j=0} \frac{(it)^j}{j!}[H^j, A] = 0 \quad \forall t\in\RR.
    \end{equation*}
    For the ($\Leftarrow$) part, suppose $[e^{iHt}, A] = 0$ for any $t$. We have,
    \begin{equation*}
        \frac{d}{dt}([e^{iHt}, A]) = iHe^{iHt}A - AiHe^{iHt} = 0.
    \end{equation*}
    Let $t = 0$, we have $e^{iHt} = I$ and hence proves the backward direction, i.e., $[H,A] = 0$.
\end{proof}

Based on the above Lemma, any linear operator $A\in \cB(\cH_d^{\ox t})$ commutes with $(e^{iHt})^{\ox t}$ must satisfy,
\begin{equation}
    \sum_{j = 0}^{t-1}[H^{(j)}, A] = 0
\end{equation}
where $H^{(j)} = I^{\ox j} \ox H \ox I^{\ox t-1-j}$. Let $t = 2$ in our cases, suppose $H = \sum_{j=0}^{d-1}E_j\ketbra{E_j}{E_j}$ be the spectral decomposition of $H$. Then for the operator $H\ox I + I\ox H$, $\{\ket{E_{jk}}:=\ket{E_j} \ox \ket{E_k}\}$ forms the eigenvectors with the corresponding eigenvalues $E_{jk}:=E_j + E_k$. Express $A$ in this basis to have,
\begin{equation*}
    [A, H\ox I + I\ox H] = \sum_{jk,mn} a_{jk,mn}(E_{jk} - E_{mn}) \ketbra{E_{jk}}{E_{mn}} = 0.
\end{equation*}
Since $a_{jk,mn} \neq 0$ in general, we derive the condition of commutation as $E_{jk} = E_{mn}$. Similarly, for $A$ commute with $e^{iHt} \ox e^{-i\overline{H} t}$ for all $t$. We have,
\begin{equation}
    [A, H\ox I - I \ox \overline{H}] = 0.
\end{equation}
Let us denote the matrix elements of $A$ in this basis $\{\ket{E_{j\overline{k}}} = \ket{E_j} \ox \ket{E^*_k}\}$. The commutation relation implies that $C_{jk, mn}$ can be non-zero only if the corresponding eigenvalues are equal, i.e., $E_j - E_k = E_m - E_n$. This is the fundamental structural constraint on $A$. It partitions the space $\mathcal{H}_d \otimes \mathcal{H}_d$ into subspaces based on the value of the energy difference $E_j - E_k$ (Bohr frequencies). The operator $A$ cannot connect these different subspaces. In other words, if we define the set of possible eigenvalue differences,
\begin{equation*}
    \Lambda = \{\mu \in \RR: \mu = E_j - E_k \text{ for }j,k\}.
\end{equation*}
Then for each $\mu$, the subspace $V_{\mu} = \spn\{\ket{E_{j\overline{k}}}: E_j - E_k = \mu\}$ is an invariant block under $H\ox I - I \ox \overline{H}.$ The commutation condition guarantees that $A = \bigoplus_{\mu\in \Lambda} A_{\mu}$ where $A_{\mu}$ acts on the subspace $V_{\mu}$.

Taking a simple example of qubit phase gate, which has been well-studied in~\cite{Vidal2002storing,Sedlak2020probabilistic}. Consider $H = Z$. Notice that $Z$ has $\pm 1$ spaces making $\Lambda = \{2, 0, -2\}$ and,
\begin{equation}
    V_{2} = \spn\{\ket{01}\}; \; V_{0} = \spn\{\ket{00}, \ket{11}\}; \; V_{-2} = \spn\{\ket{10}\};
\end{equation}

As from previous investigation, the optimal protocol $J^{\cP}$ can be constructed inspired by the decomposition of space. In the following discussion, we will prove that a class of purely coherent Lindbladians can be HPTP-programmable by constructing feasible protocols.
\begin{proposition}\label{prop:hptp_protocol_arbitrary_H_appendix}
    Let $\cL(\rho) = i[H, \rho]$ where $H$ is an arbitrary $d$-dimensional Hermitian operator with eigen-decomposition $H = \sum_{j=1}^K \lambda_j \Pi_j$ with $\Pi_j$ the orthogonal projections onto each (degenerate) eigenspaces. Then, there exists an HPTP protocol $\{\cP, \pi_t\}_t$ to exactly program $e^{\cL t}$ for any $t\geq 0$, and $d_P = K$.
\end{proposition}
\begin{proof}
    We construct $\cP$ using the projections, as follows: Define the map $\Phi_{0,1}$,
    \begin{equation*}
    \begin{aligned}
        \Phi_0(\rho_{SP}) &= (\cI_S \ox \cM)(\rho_{SP}),\; \cM(\rho_P) = (\frac{I_K}{K}\tr(\rho_P) + \rho_P  - \Delta(\rho_P))\\
        \Phi_1(\rho_{SP}) &= K \sum_{jk} (\Pi_j \ox \bra{j})(\rho_{SP})(\Pi_k \ox \ket{k}),
    \end{aligned}
    \end{equation*}
    where $\Delta$ is the completely dephasing channel with respect to the computational basis. We claim that the map $\cP:=\Phi_1\circ \Phi_0$ and the program states $\ket{\pi_t}=\sum_{j=1}^K e^{i\lambda_j t} \ket{j}$ serves as a feasible solution.

    Our first goal is to prove $\cP$ is HPTP. Clearly since $\Phi_0$ and $\Phi_1$ are HP maps, their composition $\cP$ is HP. It remains to prove that $\cP$ is TP. Notice that for any bipartite linear operator $M = \sum_{ij} a_{ij} P_i \ox Q_j$, we have,
    \begin{equation*}
    \begin{aligned}
        \Phi_0(M) &= \sum_{ij} a_{ij} P_i \ox \cM(Q_j) = \sum_{ij} a_{ij} P_i \ox  (\frac{I_K}{K}\tr(Q_j) + Q_j - \Delta(Q_j))
    \end{aligned}
    \end{equation*}
    Therefore, $\cP(M)$ reads,
    \begin{equation*}
    \begin{aligned}
        \cP(M) &= \Phi_1\left(\sum_{ij} a_{ij} P_i \ox  (\frac{I_K}{K}\tr(Q_j) + Q_j  - \Delta(Q_j))\right) \\
        &= K\sum_{ij} a_{ij} \sum_{kl} (\Pi_k \ox \bra{k})(P_i \ox  (\frac{I_K}{K}\tr(Q_j) + Q_j - \Delta(Q_j)))(\Pi_l \ox \ket{l})\\
        &= K\sum_{ij} a_{ij}\left(\frac{\tr(Q_j)}{K} \sum_{k=l} \Pi_k P_i \Pi_k + \sum_{k\neq l} \bra{k}Q_j\ket{l}\Pi_k P_i \Pi_l\right)
    \end{aligned}
    \end{equation*}
    Now, taking trace operation and using the orthogonality of $\Pi_{k}$'s to derive,
    \begin{equation*}
        \tr(\cP(M)) = \sum_{ij} a_{ij} \tr(Q_j) \tr((\sum_k\Pi_k) P_i) = \sum_{ij} a_{ij} \tr(Q_j) \tr(P_i) = \tr(M).
    \end{equation*}
    The above showcase $\cP$ is indeed TP. We remains to show that the protocol is feasible. But,
    \begin{equation*}
    \begin{aligned}
        \cP(\rho\ox \ketbra{\pi_t}{\pi_t}) &= \Phi_1(\rho \ox \ketbra{\pi_t}{\pi_t})\\
        &= \sum_{k} \Pi_k \rho \Pi_k + \sum_{k\neq l} e^{i(\lambda_k + \lambda_l)t}\Pi_k \rho \Pi_l\\
        &= \left(\sum_{k=1}^K e^{i\lambda_k t}\Pi_k\right) \rho \left(\sum_{l=1}^K e^{i\lambda_l t}\Pi_l\right)\\
        &= e^{iHt}\rho e^{-iHt},
    \end{aligned}
    \end{equation*}
    which is exactly the evolution $e^{\cL t}(\rho)$ at time $t$.
\end{proof}

This can be seen as a natural extension on the results from~\cite{Vidal2002storing}. In particular, for any given single-qubit non-trivial Hamiltonian $H$, there exists an HPTP map and corresponding program states $\pi_t$ of dimension $2$ to program the dynamics generated by $\cL = -i[H, \cdot]$.

According to the research of Jiaqing et. al.~\cite{Jiang2021physical}, we can write the SDP and its dual problem to get the cost of HPTP maps. By solving them, we can find that the cost of $\Phi$ is 2. The general programmability of coherent Lindbladians is closely related to the unitary discrimination problem. Let $U,V$ be two fixed unitary operators. The CPTP-programmability requires strong relations between them. In particular, we have,
\begin{lemma}
    Let $\cS$ be a set of unitary operations. Then, $\cS$ is CPTP-programmable if and only if $\cS=\{V|V=e^{i\phi}U,\phi\in \mathbb{R},U \text{is a fixed unitary operator}\}$.
\end{lemma}
\begin{proof}
    Consider any pure state $\psi = \ketbra{\psi}{\psi}$. Let $\cP$ is the fixed $\cptp$ programming channel and $\pi_U$ is the program state for $U$. One can always extend $\cP$ to a fixed channel $\tilde{\cP}$ s.t. $\tr_{E'}\circ \tilde{\cP} = \cP$. Then,
    \begin{equation}\label{UP1}
       \tilde{\cP}(\pi_U \ox \psi) = (U\psi U^\dagger)\otimes \pi_U' = (U\psi U^\dagger)\otimes\sum_j W_j \pi_U W_j^{\dagger}, 
    \end{equation}
    since $U\psi U^{\dagger}$ is pure. Here, $\{W_j\}_j$ is a set of Kraus operators. Suppose $\cP$ is also the programming channel for another unitary operation $V$, $\cP$ meets,
    \begin{equation}\label{eq:UP2}
       \tilde{\cP}(\pi_V \ox \psi) =(U\psi U^\dagger)\otimes\sum_j W_j \pi_V W_j^{\dagger} = V\psi V^{\dagger} \ox \pi_V'. 
    \end{equation}
    Taking partial trace over environment system of Eq.~\ref{eq:UP2}, we have
    $U\psi U^\dagger=cV\psi V^\dagger$, where $c\in\CC$ is a constant on the unit circle. So we have $V=e^{i\phi}U$.
\end{proof}

This result matches our previous findings for general CPTP-programmable Lindbladians. For any evolutionary operator $\cU_t=e^{-iHt}(\cdot)e^{iHt}$, it is CPTP-programmable if and only if $H=cI$, where $I$ is identity operator. On the other hand, another particular example of non-CPTP programmable system that has special importance in the studies chaotic behaviors of open quantum dynamics~\cite{Zanardi2021information,Styliaris2021information} is the Exchange–Dephasing Model. Consider a $2$-qubit purely coherent Lindbladian $\cL = i[S, \cdot]$ where $S$ is the SWAP gate defined as,
\begin{equation*}
    S = 
    \begin{pmatrix}
        1 & 0 & 0 & 0\\
        0 & 0 & 1 & 0\\
        0 & 1 & 0 & 0\\
        0 & 0 & 0 & 1
    \end{pmatrix}
\end{equation*}
Notice that $S$ is Hermitian and has two distinct eigenspaces spanned by $V_- = \spn_{\RR}\{\ket{\Phi_-}\}$ and $V_+ = \spn_{\RR}\{\ket{00}, \ket{11}, \ket{\Psi_+}\}$ where $\ket{\Psi_{\pm}} = \frac{1}{\sqrt{2}}(\ket{01}\pm \ket{10})$. Denote $\Pi_{\pm}$ as the projections onto the positive and negative spaces $V_{\pm}$ of $S$, respectively and the basis $\mathbb{B} = \{\ket{\Phi_{\pm}}, \ket{00}, \ket{11}\}$. Define the map $\Phi_{0,1}$ as,
\begin{equation}\label{eq:programmable_swap}
\begin{aligned}
    \Phi_0 &= \cI_A \ox \cM,\; \cM(\rho_B) = (\frac{I_2}{2}\tr(\rho_B) + \rho_B  - \Delta(\rho_B))\\
    \Phi_1 &= 2(\Pi_+ \ox \bra{0} + \Pi_- \ox \bra{1})(\cdot)(\Pi_+ \ox \ket{0} + \Pi_- \ox \ket{1}),
\end{aligned}
\end{equation}
where $\Delta$ is the completely dephasing channel with respect to the computational basis. From proposition~\ref{prop:hptp_protocol_arbitrary_H_appendix}, the map $\cP:=\Phi_1\circ\Phi_0$ is a feasible HPTP programming map for $\cL$ with program states $\ket{\pi_t} = \frac{1}{\sqrt{2}}(e^{it}\ket{0} + e^{-it}\ket{1})$.


\subsection{Proofs of programming collective dephasing two-qubit system with exchange interaction}\label{appendix:proofs_of_programming_dephasing_exchange}
In this section, we investigate a two-qubit open quantum system governed by a Lindbladian that combines coherent exchange dynamics with collective dephasing. The coherent evolution is generated by the two-qubit swap Hamiltonian, while dissipation is introduced through a collective dephasing mechanism that projects the system onto the Bell-state basis.

\begin{lemma}
    Consider $\mathcal{L} = i \, \Op{ad}_S + \lambda \left( \mathcal{D}_{\mathbb{B}} - \cI \right)$ be the Lindbladian of collective dephasing two-qubit system with exchange interaction, with $S$ is the SWAP operator and $\mathcal{D}_{\mathbb{B}}$ is the dephasing superoperator in Bell basis. The notation $\Op{ad}_S$ denotes the adjoint action of the SWAP operator $S$ on the operators, $\Op{ad}_S=[S,\cdot]$. In the case, $[\Op{ad}_S,\mathcal{D}_{\mathbb{B}}]=0$, the dynamics can be given by a sum of a time-dependent unitary and a dephasing channel $\mathcal{E}_t = e^{t \mathcal{L}}=e^{-\lambda t}e^{\Op{ad}_S}+(1-e^{-\lambda t})\mathcal{D}_{\mathbb{B}}$.
\end{lemma}
\begin{proof}
    Since $[\Op{ad}_S,\mathcal{D}_{\mathbb{B}}]=0$, the channel can be written as,
    \begin{equation*}
        \mathcal{E}_t = e^{t \mathcal{L}}= e^{i \,\Op{ad}_S + \lambda \left( \mathcal{D}_{\mathbb{B}} - \cI \right)} =  e^{i \,\Op{ad}_S}e^{\lambda t (\mathcal{D}_{\mathbb{B}}-\cI)}.
    \end{equation*}
    For dephasing channel $\mathcal{D}_{\mathbb{B}}$, any power $n >1$ satisfy $\mathcal{D}^n_{\mathbb{B}}=\mathcal{D}_{\mathbb{B}}$, so we can expand the exponential series$e^{\lambda t \mathcal{D}_{\mathbb{B}}}$,
    \begin{equation*}
        e^{\lambda t \mathcal{D}_\mathbb{B}} = \sum_{n=0}^\infty \frac{(\lambda t)^n}{n!} \mathcal{D}_\mathbb{B}^n = \cI + \left( \sum_{n=1}^\infty \frac{(\lambda t)^n}{n!} \right) \mathcal{D}_\mathbb{B} = \cI + (e^{\lambda t} - 1) \mathcal{D}_\mathbb{B}.
    \end{equation*}
    Hence,
    \begin{equation*}
        e^{\lambda t (\mathcal{D}_{\mathbb{B}}-\cI)}=e^{-\lambda t} \cI + \left(1 - e^{-\lambda t}\right) \mathcal{D}_{\mathbb{B}}.
    \end{equation*}
    Noted that $\mathcal{D}_{\text{B}}$ commutes with the unitary evolution and leaves its image invariant,
    \begin{equation*}
        \mathcal{D}_\mathbb{B} \circ e^{i \,\Op{ad}_S} = e^{i \,\Op{ad}_S}\circ \mathcal{D}_\mathbb{B}=\mathcal{D}_\mathbb{B}.
    \end{equation*}
    Therefore,
    \begin{equation*}
        \mathcal{E}_t =e^{i \,\Op{ad}_S}e^{\lambda t (\mathcal{D}_{\mathbb{B}}-\cI)} =e^{-\lambda t}e^{\Op{ad}_S}+(1-e^{-\lambda t})\mathcal{D}_{\mathbb{B}}.
    \end{equation*}
\end{proof}
\begin{proposition}\label{prop: S_dephasing_map}
    Let $\cP$ be a fixed HPTP map consists of a post-measurement process delivering outcomes $\{\ket{j}\}_{j=0}^1$ and $2$ corresponding decomposition operations $\{\cP_j\}$ (shown in Fig.~\ref{fig: SDep}). Define the program state $\pi_t = \ketbra{\theta_t}{\theta_t}\otimes\ketbra{\sigma_t}{\sigma_t}$, where the (pure) program state $\ket{\theta_t} = \frac{1}{\sqrt{2}}(e^{it}\ket{0} + e^{-it}\ket{1})$ and the control qubit $\ket{\sigma_t} = \sqrt{e^{-\lambda t}}\ket{0}+\sqrt{1-e^{-\lambda t}}\ket{1}$.
\end{proposition}

\begin{figure}[!ht]
    \centering
    \includegraphics[width=0.7\linewidth]{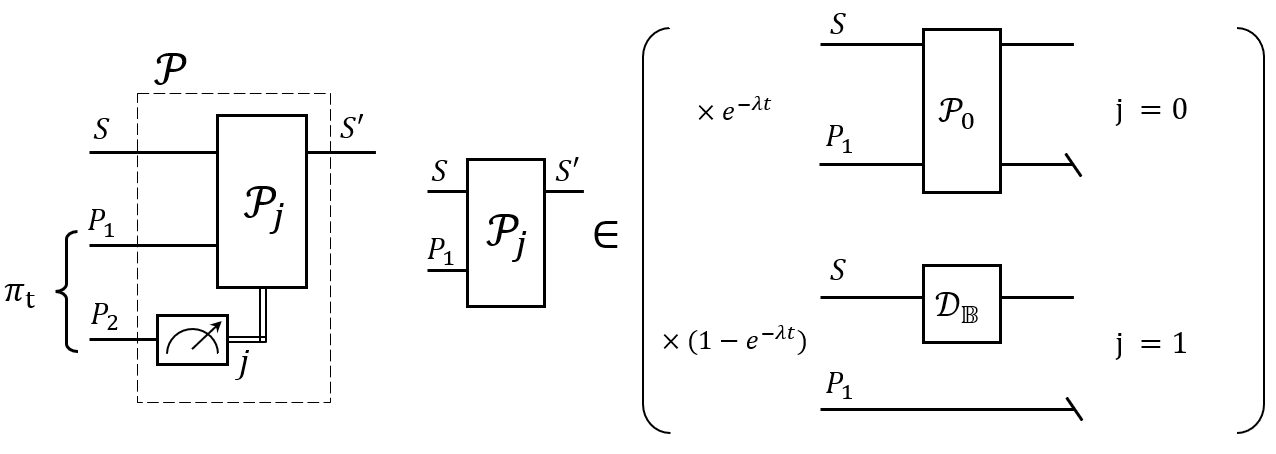}
    \caption{An $\hptp$-programming protocol of swap-dephasing channel based on $2$ measurement outcomes of the control qubit. The measurement in the programming channel is taken with respect to the computational basis $\ket{j}$ and the outcome $j$ decides which $\cP_j$ is applied. The quasi-decomposition coefficients are attached on the LHS, in front of each operation circuit.}
    \label{fig: SDep}
\end{figure}

\begin{proof}
    Let the full program space be \(R' := R\otimes C\), where \(R\) carries the original program state $\ket\theta_t$ and \(C\) is the control qubit in state of $\ket{\sigma_t}$ .
    Define
    \[
    \Pi_0 := I_R\otimes \ket{0}\bra{0},
    \qquad
    \Pi_1 := I_R\otimes \ket{1}\bra{1},
    \]
    so that \(\Pi_0+\Pi_1 = I_{R'}\) and \(\Pi_0\Pi_1 = 0\).
    Suppose the programming supermap $\mathcal{P}_1$ that always outputs the fixed channel $\mathcal{D}_{\mathbb B}$ independently of the program:
    \[
    \mathcal{P}_D(X) := \mathcal{D}_{\mathbb B},
    \qquad
    \forall\, X \in \mathcal{B}(\mathcal{H}_{R'}).
    \]
    Since $\mathcal{D}_{\mathbb B}$ is CPTP, the associate map $\mathcal{P}_1$ is CPTP.
    We now define the new programming supermap,
    \begin{equation*}
      \mathcal{P}(M)
      \;:=\;
      \mathcal{P}_0\!\big(\Pi_0 M \Pi_0\big)
      \;+\;
      \mathcal{P}_1\!\big(\Pi_1 M \Pi_1\big),
      \qquad
      M\in\mathcal{B}(\mathcal{H}_{R'}),
    \end{equation*}
    where $\mathcal{P}_0$ is the same programming map as Lemma \ref{eq:programmable_swap}, and therefore $\mathcal{P}_0$ is HPTP. For \(M = M^\dagger\), \(\Pi_k M \Pi_k\) is Hermitian for \(k=0,1\). Hence, both $\mathcal{P}_0$ and $\mathcal{P}_1$ in $\mathcal{P}$ are HPTP, so $\mathcal{P}$ is HPTP.
    The program state $\pi_t$ can be written as,
    \begin{equation*}
        \pi_t = \theta_t \otimes \sigma_t
        =
        e^{-\lambda t}\theta_t\otimes\ket{0}\bra{0}
        +
        (1-e^{-\lambda t})\theta_t\otimes\ket{1}\bra{1}
        + \text{(off-diagonal terms)}.
    \end{equation*}  
    Hence,
    \[
    \Pi_0\pi_t\Pi_0
    =
    e^{-\lambda t}\theta_t\otimes\ket{0}\bra{0},
    \qquad
    \Pi_1\pi_t\Pi_1
    =
    (1-e^{-\lambda t})\theta_t\otimes\ket{1}\bra{1}.
    \]
    Therefore, for any $\rho$,
    \begin{equation*}
    \begin{aligned}
        \mathcal{P}(\rho \otimes \pi_t) & = \mathcal{P}_0\!\big(\rho \otimes\Pi_0 \pi_t \Pi_0\big)+
        \mathcal{P}_1\!\big(\rho \otimes\Pi_1 \pi_t \Pi_1\big)\\
        & = e^{-\lambda t}\mathcal{P}_0\!\big(\rho \otimes\pi_t \otimes \ketbra{0}{0})+
        (1-e^{-\lambda t})\mathcal{P}_1\!\big(\rho \otimes\pi_t \otimes \ketbra{1}{1})\\
        & = e^{-\lambda t}e^{\Op{ad}_S}(\rho)+(1-e^{-\lambda t})\mathcal{D}_{\mathbb{B}} (\rho)\\
        & = \mathcal{E}_t (\rho)
    \end{aligned}
    \end{equation*}
\end{proof}

\begin{lemma}\label{lem: swap_dephasing_curve}
    Consider the initial state $\ket{\psi_0}=\ket{01}=\frac{1}{\sqrt{2}}(\ket{\Psi^{+}}+\ket{\Psi^{-}})$ be the ground state of the collective dephasing two-qubit open quantum system with exchange interaction, where $\ket{\Psi^{\pm}}$ is the Bell state $\ket{\Psi^{\pm}} = \frac{1}{\sqrt{2}}(\ket{01}\pm\ket{10})$. For the Lindbladian $\mathcal{L} = i \, \Op{ad}_S + \lambda \left( \mathcal{D}_{\mathbb{B}} - \cI \right)$, the comparison of the evolution of the initial state overlap $\bra{\psi_0}\rho(t)\ket{\psi_0}=\frac12[1 + e^{-\lambda t}\cos(2t)]$.
\end{lemma}
\begin{proof}
    The SWAP operator has the following eigenstates with eigenvalues,
    \begin{equation*}
        S|\Psi_+\rangle = +|\Psi_+\rangle,\qquad
        S|\Psi_-\rangle = -|\Psi_-\rangle.
    \end{equation*}
    Denote that $\alpha$ and $\beta$ are indices that label the Bell eigenstates of the SWAP. Hence,
    \begin{equation*}
        \alpha,\beta \in \{+, -\},\quad |\alpha\rangle \in {|\Psi_+\rangle,|\Psi_-\rangle},
    \end{equation*}
    and $\rho_{\alpha\beta} = \langle \alpha|\rho|\beta\rangle$ are the corresponding matrix elements of $\rho$ in this eigenbasis,
    \begin{equation*}
        \rho(t) =
        \begin{pmatrix}
        \rho_{++}(t) & \rho_{+-}(t)\\
        \rho_{-+}(t) & \rho_{--}(t)
        \end{pmatrix}.
    \end{equation*}
    At $t=0$, 
    \begin{equation*}
        \rho(0) =
        \begin{pmatrix}
        \rho_{++}(0) & \rho_{+-}(0)\\
        \rho_{-+}(0) & \rho_{--}(0)
        \end{pmatrix}=
        \ketbra{01}{01}=
        \frac{1}{2}
        \begin{pmatrix}
            1&1\\
            1&1
        \end{pmatrix}
    \end{equation*}
    The dephasing channel $\mathcal{D}_{\mathbb{B}}$ acts by killing all off-diagonal elements in this eigenbasis,
    \begin{equation*}
        \mathcal{D}_{\mathbb{B}}(\rho) = \sum_{\alpha} P_\alpha \rho P_\alpha,
        \quad P_i = \ketbra{\alpha}{\alpha},
    \end{equation*}
    Denote that $\lambda_\alpha$ and $\lambda_\beta$ are the eigenvalues of $S$ with eigenstates $\ket{\alpha}$ and $\ket{\beta}$. In the eigenbasis of $S$, the commutator has matrix elements,
    \begin{equation*}
        (\Op{ad}_S(\rho))_{\alpha\beta}
   = (S\rho - \rho S)_{\alpha\beta}
   = (\lambda_\alpha - \lambda_\beta)\rho_{\alpha\beta}.
    \end{equation*}
    We can now write the equation and solve for the diagonal and off-diagonal entries,
    \begin{equation*}
        \dot{\rho} = i[S,\rho] + \lambda(\mathcal{D}_{\mathbb{B}}(\rho) - \rho)
    \end{equation*}
    For the diagonal entries, $\alpha=\beta$, the commutator has zero diagonal in the eigenbasis,
    \begin{equation*}
        (\Op{ad}_S(\rho))_{\alpha\alpha} = (\lambda_\alpha-\lambda_\alpha)\rho_{\alpha\alpha} = 0,
    \end{equation*}
    while dephasing $\mathcal{D}_{\mathbb{B}}$ does nothing to the diagonal,
    \begin{equation*}
        \frac{d}{dt}\rho_{\alpha\alpha}== i(\Op{ad}_S(\rho))_{\alpha\alpha}
+\lambda\bigl[(\mathcal{D}_{\mathbb{B}}(\rho))_{\alpha\alpha}-\rho_{\alpha\alpha}\bigr]
  = 0,
    \end{equation*}
    so $\rho_{\pm\pm}(t) = \rho_{\pm\pm}(0) = \frac12$.
    For the off-diagonal entries, $\alpha\neq\beta$, 
    \begin{equation*}
        \frac{d}{dt}\rho_{\alpha\beta}
        = i(\lambda_\alpha-\lambda_\beta)\rho_{\alpha\beta}
        +\lambda(0 - \rho_{\alpha\beta})= \bigl[i(\lambda_\alpha-\lambda_\beta) - \lambda\bigr]\rho_{\alpha\beta}.
    \end{equation*}
    With eigenvalues $\lambda_\pm=\pm1$, the solution of this ordinary differential equation is,
    \begin{equation*}
    \begin{aligned}
        &\rho_{+-}(t) = \rho_{+-}(0)e^{(2i - \lambda)t} = \frac12 e^{(-\lambda + 2i)t},\\
        & \rho_{-+}(t) = \rho_{-+}(0)e^{(-2i - \lambda)t} = \frac12 e^{(-\lambda -2i)t}.
    \end{aligned}
    \end{equation*}
    Hence,
    \begin{equation*}
        \rho(t) =
        \begin{pmatrix}
        \rho_{++}(t) & \rho_{+-}(t)\\
        \rho_{-+}(t) & \rho_{--}(t)
        \end{pmatrix}= \frac{1}{2}
        \begin{pmatrix}
            1 & e^{-\lambda t}e^{i2t}\\
            e^{-\lambda t}e^{-i2t} & 1 
        \end{pmatrix}.
    \end{equation*}
    We already know the coordinates of $|01\rangle$ in basis of $\Psi_\pm$ is,
    \begin{equation*}
        |01\rangle \ \leftrightarrow\ v = \frac{1}{\sqrt2}\begin{pmatrix}1\\1\end{pmatrix}.
    \end{equation*}
    Hence,
    \begin{equation*}
        \bra{\psi_0}\rho(t)\ket{\psi_0}=v^{\dagger}\rho(t)v=\frac1{4}\left[2 + e^{-\lambda t}(e^{i2t}+e^{-i2t})\right]. 
    \end{equation*}
    Recall that $e^{i2t}+e^{-i2t} = 2\cos(2t)$, therefore,
    \begin{equation*}
        \bra{\psi_0}\rho(t)\ket{\psi_0}= \frac12\left[1 + e^{-\lambda t}\cos(2t)\right]
    \end{equation*}
\end{proof}

\subsection{Proofs of programming photon loss system}\label{appendix:proofs_of_photonic_loss}

In this section, we take a step to investigate a purely dissipative Lindbladian that is also not CPTP-programmable. The photonic loss systems are essential for understanding how dissipation and decoherence shape light–matter interactions, enabling accurate modeling of open quantum systems and non-Hermitian dynamics. Insights into engineered loss inform the design of robust photonic devices by optimizing performance, stability, and noise resilience. We consider the Lindbladian characterized by a single jump operator $\sqrt{\gamma}\ketbra{0}{1}$ corresponding to the annihilation process of rate $\gamma$. The following results can be established.
\begin{lemma}
    Let $\cL$ be the Lindbladian of a qubit photon loss system with one jump operator $L = \sqrt{\gamma} \ketbra{0}{1}$. Then, $\cL$ is not CPTP-programmable.
\end{lemma}
\begin{proof}
    Suppose the damping rate $\gamma > 0$. For a qubit with basis ${|0\rangle, |1\rangle}$, if we write the action of $\cL$ as 
    \begin{equation*}
        \cL(\rho) = \cL\left(\begin{pmatrix} \rho_{00} & \rho_{01} \\ \rho_{10} & \rho_{11} \end{pmatrix}\right) = \gamma
        \begin{pmatrix} \rho_{11} & -\frac{1}{2}\rho_{01} \\ -\frac{1}{2}\rho_{10} & -\rho_{11} \end{pmatrix},
    \end{equation*} 
    We analyze $\mathcal{L}$ as a linear operator on the space of $2 \times 2$ density matrices. 
    Now suppose there exist $\alpha > 0$ and a quantum channel $\cE$ such that
    \begin{equation*}
    \cL = \alpha(\cE - \cI).
    \end{equation*}
    We can rearrange the assumed equality to solve for the map $\cE(\rho) = \rho + \frac{1}{\alpha}\cL(\rho)$ so that,
    \begin{equation*}
        \cE(\rho) = 
        \begin{pmatrix} \rho_{00} + \frac{\gamma}{\alpha}\rho_{11} & (1-\frac{\gamma}{2\alpha})\rho_{01} \\ (1-\frac{\gamma}{2\alpha})\rho_{10} & (1- \frac{\gamma}{\alpha})\rho_{11} \end{pmatrix}.
    \end{equation*}
    For $\cE$ to be a quantum channel, it must be a CPTP linear map, or equivalent to say $J(\cE) \geq 0$ and $\tr_2(J(\cE)) = I_2$. 
    Now, we construct the $4 \times 4$ Choi matrix $J(\mathcal{E})$ in the basis ${|00\rangle, |01\rangle, |10\rangle, |11\rangle}$:
    \begin{equation*}
    J(\cE) = \sum_{ij} \ketbra{i}{j} \ox \cE(\ketbra{i}{j}) = 
    \begin{pmatrix}
        1 & 0 & 0 & 1-\frac{\gamma}{2\alpha}\\
        0 & 0 & 0 & 0\\
        0 & 0 & \frac{\gamma}{\alpha}& 0\\
        1 - \frac{\gamma}{2\alpha} & 0 & 0 & 1 - \frac{\gamma}{\alpha}
    \end{pmatrix}
\end{equation*}
For this matrix to be positive semidefinite, all of its principal minors must be non-negative. Setting $\eta = \gamma / \alpha$, we reorder the basis to $\{|00\rangle, |11\rangle, |01\rangle, |10\rangle\}$ to make the block structure,
$$ J'(\mathcal{E}) = \begin{pmatrix} 1 & 1-\frac{\eta}{2} & 0 & 0 \\ 1-\frac{\eta}{2} & 1-\eta & 0 & 0 \\ 0 & 0 & 0 & 0 \\ 0 & 0 & 0 & \eta \end{pmatrix} 
$$ 
This is a block-diagonal matrix. For $J'(\mathcal{E})$ to be positive semidefinite, each block on the diagonal must be positive semidefinite. The two $2 \times 2$ block on the right-down corner is trivially positive semidefinite. It remains to check the $2 \times 2$ left up-corner block $M$ must be positive semidefinite. However, its determinant
\begin{equation*}
    \det(M) = (1-\eta) - (1-\frac{\eta}{2})^2 = -\frac{\eta^2}{4} < 0,
\end{equation*}
as $\eta > 0$, which contradicts the CP condition of $\cE$, and therefore, $\cL$ is not CPTP-programmable.
\end{proof}

Denote $\cE^{\Op{AD}}_{\gamma}$ to be the amplitude damping channel where $\cE^{\Op{AD}}_{\gamma}(\rho) = E_0\rho E_0^{\dagger} + E_1 \rho E_1^{\dagger}$ with,
\begin{equation*}
    E_0 = 
    \begin{pmatrix}
        1 & 0\\
        0 & \sqrt{1-\gamma}
    \end{pmatrix};\quad
    E_1 = 
    \begin{pmatrix}
        0 & \sqrt{\gamma}\\
        0 & 0
    \end{pmatrix}
\end{equation*}
From~\cite{Nielsen2010quantum}, one can simulate amplitude damping channel via a simple parameterized quantum circuit as shown in Fig.~\ref{fig:amp} where $\theta = 2\arcsin{(\sqrt{\gamma})}$. Notice that the $\Op{CR}_Y$ gate can be decomposed into 
\begin{equation}\label{eq:decomp_cry}
    \Op{CR}_Y(\theta)=\Op{CNOT}\left(I\otimes R_Y\left(-\frac{\theta}{2}\right)\right)\Op{CNOT}\left(I\otimes R_Y\left(\frac{\theta}{2}\right)\right)
\end{equation}
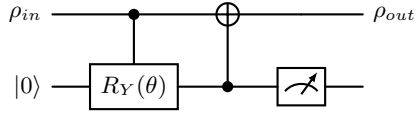
\begin{figure}[h!]
    \centering
    \begin{quantikz}[transparent]
    \lstick{$\rho_{in}$} & \ctrl[]{1} & \targ{} & {} & \rstick{$\rho_{out}$} \\
    \lstick{$\ket{0}$} & \gate[1]{R_Y(\theta)} & \ctrl[]{-1} & \meter{} & \\
    \end{quantikz}
    \caption{Circuit implementation  for realizing the amplitude damping channel.}
    \label{fig:amp}
\end{figure}
As from the studies on circuit cutting~\cite{Piveteau2023circuit,Jing2025circuit,Mitarai2021constructing}, one can decompose $\Op{CNOT}$ operation, denoted as curly $\cC\cX$ into affine combination of product operations:
\begin{equation}
    \cC\cX = \Pi_0 \ox \cI + \Pi_1 \ox \cX + \cI \ox \Pi_+ + \cZ \ox \Pi_- - \frac{1}{2}\cS \ox \hat{\cH}\cS^\dagger\hat{\cH} - \frac{1}{2}\cS^\dagger \ox \hat{\cH}\cS\hat{\cH}, 
\end{equation}
where $\cX(\rho) = X\rho X$, $\cZ(\rho) = Z\rho Z$, $\cS(\rho) = S\rho S^\dagger$, $\cS^\dagger(\rho) = S^\dagger \rho S$, and $\hat{\cH}(\rho) = H\rho H$, are the single-qubit common gate operations; $\Pi_{\pm}(\rho) = \ketbra{\pm}{\pm}X\ketbra{\pm}{\pm}$ are the projection operation onto the states $\ket{\pm}$. Each of the decomposition operation is attached with a coefficient $\alpha_j$ with the corresponding sign $\Op{sgn}(\alpha_j)$. The total sampling overhead of realizing the operation is computed as the summation of all the decomposition coefficients $\kappa = 4+2\times (1/2) = 5$. Combine with the quasi-sampling method, we propose an $\hptp$-programming protocol for the amplitude damping channel with the details in the following proposition.

\begin{proposition}
    Let $\cP$ be a fixed $\hptp$ map consists of a post-measurement process delivering outcomes $\{\ket{j}\}_{j=1}^6$ and $6$ corresponding decomposition operations $\{\cP_j\}$ (shown in Fig.~\ref{fig:decom}). Define the program states $\pi_{\theta} = \sum_{j=1}^6 p_j\sigma_j(\theta) \ox \ketbra{j}{j}_{P_2}$ where the sampling probability $p_j = |\alpha_j| / \kappa$, the (pure) states $\sigma_j$'s are,
    \begin{equation*}
    \begin{cases}
        \ket{\sigma_1(\theta)} = \ket{0},\\
        \ket{\sigma_2(\theta)} = \sin\left(\frac{\theta}{2}\right)\ket{0} + \cos\left(\frac{\theta}{2}\right)\ket{1},\\
        \ket{\sigma_3(\theta)} =\frac{1}{\sqrt{2}}(\sqrt{1+\sin(\theta / 2)} \ket{0} + \frac{\cos(\theta / 2)}{\sqrt{1+\sin(\theta/2)}}\ket{1}),\\
        \ket{\sigma_4(\theta)} =\frac{1}{\sqrt{2}}(\sqrt{1-\sin(\theta / 2)} \ket{0} - \frac{\cos(\theta / 2)}{\sqrt{1-\sin(\theta/2)}}\ket{1}),\\
        \ket{\sigma_5(\theta)} =  \frac{1}{2}(((1-i) + \sin(\theta / 2)(1+i))\ket{0} + (1+i)\cos(\theta / 2)\ket{1}),\\
        \ket{\sigma_6(\theta)} =  \frac{1}{2}(((1+i) + \sin(\theta / 2)(1-i))\ket{0} + (1-i)\cos(\theta / 2)\ket{1}),
    \end{cases}
    \end{equation*}
    where $\sigma_j = \ketbra{\sigma_j}{\sigma_j}$. Then, by setting $\theta = 2\arcsin{(\sqrt{\gamma})}$, the programming channel $\cP$ defined in Fig.~\ref{fig:decom} can exactly simulate $\cA_{\gamma}$.
\end{proposition}

\begin{figure}[!ht]
    \centering
    \includegraphics[width=1.0\linewidth]{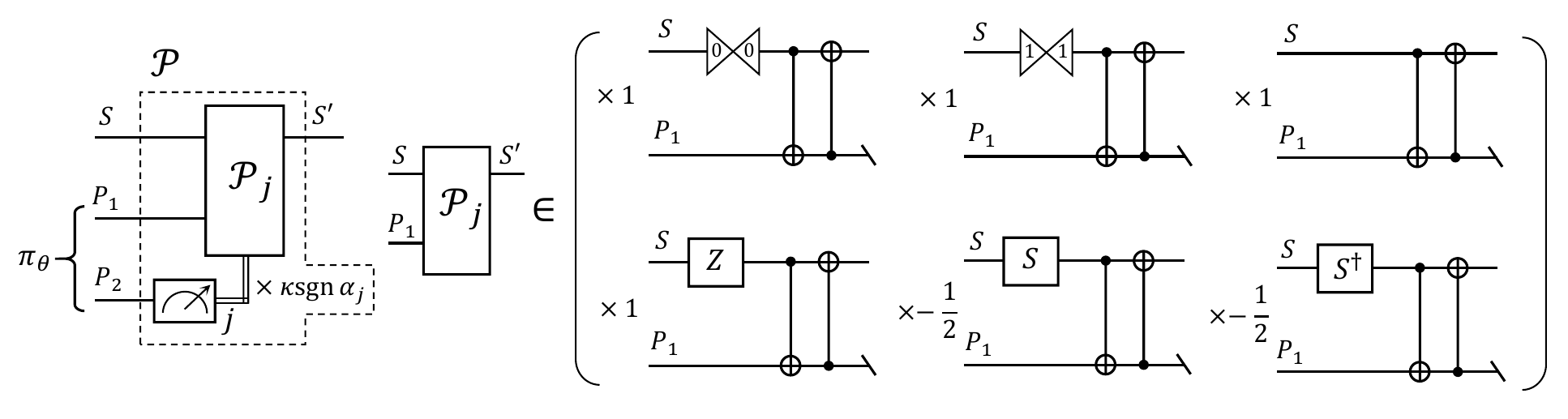}
    \label{fig:decom}
    \caption{An $\hptp$-programming protocol of amplitude damping channel based on $6$ decomposition product operations from circuit cutting of $\Op{CNOT}$ gate. The measurement in the programming channel is taken with respect to the computational basis $\ket{j}$ and the outcome $j$ decides which $\cP_j$ is applied. The quasi-decomposition coefficients are attached on the LHS, in front of each operation circuit.}
\end{figure}

\begin{proof}
    We start with the derivation of $\ket{\sigma_j}$'s using the gate decomposition~\eqref{eq:decomp_cry}. These states is given by applying the $R_Y(-\theta / 2) \cE_j R_Y(\theta/2)$ to $\ket{0}$ where $\cE_j$ are the local operations from the quasi-decomposition of $\cC\cX$. States $\sigma_{1,2,5,6}$ are the normal parametrized (pure) quantum states. While $\sigma_{3,4}$ are originally unormalized states from applying the completely positive and trace-nonincreasing (CPTN) operation $R_Y(-\theta / 2) \Pi_{\pm} R_Y(\theta/2)$. These can be physically implemented by extending the CPTN operation to CPTP. Combine with the measurements and the classical post-measurement processes, we derive,
    \begin{equation*}
    \begin{aligned}
        \cP(\rho \ox \pi_{\theta}) &= \kappa \sum_j p_j\Op{sgn}(\alpha_j) (\cI\ox \cR_Y(-\theta / 2))\circ\cP_j\circ(\cI\ox \circ\cR_Y(\theta / 2))(\rho\ox \sigma_j(\theta))\\
        &=\sum_j \alpha_j(\cI\ox \cR_Y(-\theta / 2))\circ\cP_j\circ(\cI\ox \circ\cR_Y(\theta / 2))(\rho\ox \sigma_j(\theta))\\
        &= \tr_{P_1}(\cC\cR_Y(\theta) \circ \cC\cX_{P_1\rightarrow S}(\rho\ox \ketbra{0}{0}))\\
        &= \cA_{\gamma}(\rho),
    \end{aligned}
    \end{equation*}
    where the last equality holds due to the implementation in Fig.~\ref{fig:amp}.
\end{proof}

\section{Estimating programming cost of open-system dynamics}\label{appendix:programmability of oqs}

The programmability of open quantum systems raises from the physical implementability and the simulation cost of general linear maps~\cite{Regula2021operational}. We target to construct a physically realizable processor that can simulate the dynamics $\cA_t$ generated by a fixed Lindbladian $\cL$ at any time $t\geq 0$. In particular, Let $\cH_{tot} = \cH_S \ox \cH_P$ be a composite quantum system. A programmable quantum processor integrates a HPTP linear map $\cP\in \hptp(\cH_{tot} \rightarrow \cH_{S'})$ with $\cH_{S'} \simeq \cH_{S}$, and a continuous set of program states $\pi_t$ acting on $\cH_P$.
\begin{definition}[Programmability of Lindbladian]
    Let $\cL$ be a Lindbladian and denote $\Omega=\hptp(\cH_{tot} \rightarrow \cH_{S'})$. For some $T\geq 0$, we say $\cL$ is $\Omega_{\epsilon}$-programmable within the time interval $[0,T]$ if there exists an one-parameter, continuous set of program states $\pi_t \in \cD(\cH_P)$ and a $\cP\in\Omega$ such that $\forall t\in [0,T]$, 
    \begin{equation*}
        \frac{1}{2}\|\cP(\cdot \ox \pi_t) - e^{t \cL} \|_{\diamond} \leq \epsilon.
    \end{equation*}
    We call $\{\cP, \pi_t\}_t$ a $\Omega_{\epsilon}$-programming protocol of $\cL$ within $[0,T]$, and $\cP$ the programming channel, (or retrieval channel). If such a protocol exists as $T\rightarrow \infty$, we say $\cL$ is $\Omega_{\epsilon}$-programmable.
\end{definition}
Specifically, when $\epsilon = 0$, we call it the (exact) $\Omega$-programmable scenario, and omit writing the $\epsilon$ subscript. The programming cost of $\cL$ is defined as follows:
\begin{definition}[Physical programming cost]\label{def:cost_vir} Given $T\geq 0$, a continuous set of program states $\{\pi_t\}$ defined on $[0,T]$, and $\epsilon\geq 0$, the $\epsilon$-error \textit{programming cost} of the superoperator $\cL$ with respect to the program states $\pi_t$ in the time interval is defined as
\begin{equation}\label{eq:def_programming_cost}
\begin{aligned}
    \gamma_{\epsilon}(\pi_t,\cL, T):=\log \min_{\cP \in \Omega}\left\{||\cP||_\diamond\Bigg\arrowvert \forall t\in [0,T], \frac{1}{2}\left\|\cP(\cdot \otimes\pi_t) - e^{t \cL}\right\|_{\diamond} \leq \epsilon\right\}.
\end{aligned}
\end{equation}
\end{definition}
There is no guarantee for the cost to be finite. Numerically, most choices of $\pi_t$ can leads to a infinite cost value. We say $\pi_t$ is \textit{veritable} if $\gamma_{\epsilon}(\pi_t, \cL, T) < \infty$. Clearly, if there exists such program states for $\cL$ such that the cost is strictly bounded above, $\cL$ is $\Omega_{\epsilon}$-programmable within $[0,T]$. 
\begin{remark}
    For any $\epsilon \geq 0$, Lindbladian $\cL$ with veritable program states $\{\pi_t\}$ within the interval $[0,T]$, we have,
    \begin{enumerate}
        \item \textbf{(Initial cost)} $\gamma_{\epsilon}(\pi_t, \cL, 0) = 0$.
        \item \textbf{(Monotonicity in $T$)} Let $T_2 \geq T_1 \geq 0$, for any $\epsilon \geq 0$,  Lindbladian $\cL$ with veritable $\pi_t$, we have $\gamma_{\epsilon}(\pi_t, \cL, T_2) \geq \gamma_{\epsilon}(\pi_t, \cL, T_1)$.
    \end{enumerate}
\end{remark}
\begin{proof}
    For the initial cost, given $\epsilon$, $\cL$, and $\pi_t$, $e^{0 \cdot \cL} = \cI$, construct $\cP = \tr_{P}$, so that, $\cP(\rho \ox \pi_0) = \cI(\rho)$. Hence, $0 \leq \gamma_{\epsilon}(\pi_t, \cL, 0) \leq 0$; The monotonicity directly follows by the property of the optimization problem. With $T_2\geq T_1$, the solution space is narrowed. 
\end{proof}

Since $\gamma_{\epsilon}(\pi_t, \cL, T)$ is monotone in $T$, if the programming cost can be bounded above by taking $T\rightarrow \infty$, the entire semigroup $e^{t \cL}$ can be $\Omega_{\epsilon}$-programmable, and we can switch the order of limit and minimization to define,
\begin{equation}
    \gamma_{\epsilon}(\pi_t,\cL) :=\lim_{T\rightarrow \infty}\gamma_{\epsilon}(\pi_t,\cL, T)
\end{equation}

Notice that the program states in the definition of programmability is defined by a continuous parameter $t$. In general, this can cause theoretical hardness on analyzing the trajectory of $\pi_t$. Within an infinitesimal time difference, there may exist multiple veritable program states $\pi_t$ and $\pi_{t+\delta t}$, which induces extreme `zigzag' patterns in the $\cD(\cH_d)$. In such cases, preparing $\{\pi_t\}_t$ can be physically impossible. To avoid this, we have made the assumption that the program states are defined analytically in $t$.

\begin{theorem}[Identity theorem~\cite{Ahlfors1979complex,Krantz2002primer}]
    Given functions $f$ and $g$ analytic on a domain $D$ (open and connected subset of $\mathbb{R}$ or $\mathbb {C}$), if $f = g$ on some $S\subseteq D$, where $S$ has an accumulation point in $D$, then $f = g$ on $D$.
\end{theorem}

\begin{lemma}\label{lem:analytic_continuation_protocol}
    Let $\pi_t$ be analytically defined on $[0,\infty)$, and $\cP$ is some exact programming channel of $\cL$ within the time interval $[0,T]$ for some $T>0$. Then, $\{\cP, \pi_t\}_t$ serves an exact programming protocol for any $t \geq T$.
\end{lemma}
\begin{proof}
    Let us define two functions, $\Phi_t = e^{t \cL}$ and $\Psi_t = \cP(\cdot \ox \pi_t)$, which map non-negative real numbers to the space of superoperators. By definition of $e^{t \cL}$, the power series converges in the operator norm for $t\in \RR$, and hence, is analytic. Now, since $\cP$ is a fixed linear map, write $\pi_t$ as a convergent power series,
    \begin{equation*}
        \pi_t = \sum_{k=0}^{\infty} A_k (t - t_0)^k,
    \end{equation*}
    at any $t_0\in[0,\infty)$, so that for any $\rho\in \cD(\cH_S)$,
    \begin{equation*}
        \cP(\rho \ox \pi_t) = \sum_{k=0}^{\infty} \cP(\rho \ox A_k) (t - t_0)^k = \sum_{k=0}^{\infty} \cP_k(\rho) (t - t_0)^k
    \end{equation*}
    where $\cP_k(\rho) := \cP(\rho \ox A_k)$. Notice that $\cP_k$ are linear maps so that $\cP(\cdot \ox \pi_t)$ is defined as the composition and the power series of linear maps, therefore is analytic. Now, since $\Phi_t = \Psi_t$ for all $t$ in the interval $[0, T]$. According to the identity theorem, we have $\Phi_t = \Psi_t$ for any $t\geq T$.
\end{proof}

This Lemma showcases a counterintuitive fact that if we can exactly and deterministically simulate the dynamics of $e^{t \cL}$ within a small time interval $[0,T]$. Then, the mathematical fundamentals forces $\cP(\cdot\ox \pi_t) = e^{t \cL}$ for any $t$ extending to infinity On the other hand, this provides a physical convenience that if we can construct a perfect programming channels by injecting analytic $\pi_t$ within a short time interval. Then, it can characterize the exact programming cost as $T\rightarrow \infty$, i.e., $\gamma(\pi_t, \cL, T) = \gamma(\pi_t, \cL)$.

By setting specific set of program states $\pi_t$, one can estimate the (exact) cost value via solving the following optimization problem~\eqref{eq:sdp_virtual_cost}.
\begin{equation}\label{eq:sdp_virtual_cost}
\begin{aligned}
    2^{\gamma(\pi_t, \cL)}&=\min\;p_1+p_2\\
    {\rm s.t.}\;\;& J^{\cP}:=J_1-J_2,\\
    &\tr_{P}[J^{\cP}(\pi_t^{T} \ox I_{SS'})] = J(e^{t \cL})_{SS'},\, \forall t \geq 0,\\
    &J_1\geq 0,\,\tr_{S'}[J_1] =p_1 I_{SP},\\
    &J_2\geq 0,\,\tr_{S'}[J_2] =p_2 I_{SP},
\end{aligned}
\end{equation}
where $J_{1,2}$ and $p_{1,2}$ are the optimization variables $d$ is the dimension of $\cH_S$ and the subscript alphabet on the operators represent the acted corresponding systems. By sampling a sufficient number of time steps, the cost can be estimated via solving a semidefinite programming. The transformation between the Lindbladian and the Choi representation can be exactly computed using the formulas derived in Appendix~\ref{appendix:lindblad_to_choi}.

\subsection{Properties of programming cost}\label{appendix:properties_cost}

In this section, we have assumed to use analytic program states to remove the $T$ dependence in the programming cost. We will first state some basic properties and definitions of Lindbladians and the cost for a further discussion.

\begin{definition}[Commutator of linear maps]
    Let $\cA$ and $\cB$ be two linear maps acting on $\cB(\cH_d)$. The commutator of $\cA,\cB$ is a linear map defined as,
    \begin{equation*}
        [\cA,\cB](X) := \cA\circ \cB(X) - \cB\circ \cA(X) = (\cA \circ \cB - \cB\circ \cA)(X).
    \end{equation*}
    We say $\cA$ and $\cB$ commute if for any $X\in\cB(\cH_d)$, $[\cA, \cB] = \bm{0}$.
\end{definition}

\begin{remark}
    Given $\cL_1$, $\cL_2$ two physical Lindbladians acting on $\cD(\cH)$. Then, the following operations can make valid Lindbladians
    \begin{enumerate}
        \item For any real numbers $a,b\geq 0$, $a\cL_1 + b\cL_2$ is also a physical Lindbladian.
        \item Let $\cI$ denote the identity map acting on $\cD(\cH)$, then, $\cL_1\ox \cI + \cI \ox \cL_2$ is a physical Lindbladian.
    \end{enumerate}
\end{remark}

\begin{lemma}\label{lem:bound_approx_sum_of_lind}
    Let $\cL$ be some $\Op{HPTP}_{\epsilon}$-programmable Lindbladian with program states $\{\pi_t\}_t$ for $\epsilon \geq 0$. Then, for any fixed $a>0$, there exist program states $\{\omega_t\}_t$ s.t., $\gamma_{\epsilon}(\pi_t, a\cL) = \gamma_{\epsilon}(\omega_t, \cL)$.
\end{lemma}
\begin{proof}
    Let $\cP$ be some optimal $\epsilon$-error programming channel for $\cL$ w.r.t. $\pi_t$. Setting $\tau = t/a$, $\omega_{\tau} := \pi_{t}$, we have,
    $$\frac{1}{2}\|\cP(\cdot \ox \omega_{\tau}) - e^{a\cL \tau} \|_{\diamond} = \frac{1}{2}\|\cP(\cdot \ox \pi_{t}) - e^{a\cL t/a} \|_{\diamond} \leq \epsilon$$
    for all $\tau\geq 0$, so that $\{\cP, \omega_{\tau}\}_{\tau}$ forms a feasible solution, and $\gamma_{\epsilon}(\pi_t, a\cL) \leq \gamma_{\epsilon}(\omega_t, \cL)$. One the other hand, Let $\cQ$ be some optimal $\epsilon$-error programming channel for $a\cL$ w.r.t. $\omega_t$. Setting $\tau = at$, $\pi_\tau:=\omega_t$ so that $\cQ$ can form a feasible solution for $\cL$. Therefore, $\gamma_{\epsilon}(\pi_t, a\cL) \geq \gamma_{\epsilon}(\omega_t, \cL)$
\end{proof}

This Lemma showcases that the (real) scalar multiplication on the Lindbladian does not cause differences in the programming cost. The definition of the programmability works for any time $t\geq 0$ which absorbs the scalar effect in the programming process. As a result, we will discuss the properties of the programming cost by ignoring the scalar multiplication on the Lindbladian in the following sections.
\begin{lemma}
    Let $\cL$ be some $\Op{HPTP}_{\epsilon}$-programmable Lindbladian with program states $\{\pi_t\}_t$, and $\cE$ be some invertible quantum channel acting on $\cH_P$. Then, $\gamma(\pi_t, \cL) \leq \gamma(\cE(\pi_t), \cL) \leq \gamma(\pi_t, \cL) + \log\| \cE^{-1}\|_{\diamond}$.
\end{lemma}
\begin{proof}
    Let $\cP$ be an optimal solution to the optimization problem so that $2^{\gamma(\pi_t, \cL)} = \|\cP\|_{\diamond}$. Now for the new program states $\{\cE(\pi_t)\}_t$, define  $\Tilde{\cP}:=\cP\circ(\cI\ox \cE^{-1})$ so that, for any $\rho_0 \in \cD(\cH_d)$,
    \begin{equation*}
        \Tilde{\cP}(\rho_0 \ox \cE(\pi_t)) = \cP(\rho_0 \ox \cE^{-1} \circ \cE(\pi_t)) = \cP(\rho_0 \ox \pi_t) = e^{t \cL}(\rho_0),
    \end{equation*}
    which showcases that $\Tilde{\cP}$ serves as a feasible solution for $\cE(\pi_t)$. Since $\cE^{-1}$ is a HP unital map, using the subadditivity w.r.t the map composition~\cite{Jiang2021physical},
    $$2^{\gamma(\cE(\pi_t), \cL)} \leq \|\cP\circ(\cI\ox \cE^{-1})\|_{\diamond} \leq  \|\cP\|_{\diamond}\cdot\| \cE^{-1}\|_{\diamond} \leq 2^{\gamma(\pi_t, \cL)}\| \cE^{-1}\|_{\diamond}$$ 
    which implies $\gamma(\cE(\pi_t), \cL) \leq \gamma(\pi_t, \cL) + \log\| \cE^{-1}\|_{\diamond}$. Using similar logic, we can construct a feasible solution for $\pi_t$ by assuming the optimality from $\cE(\pi_t)$, and hence, $\gamma(\cE(\pi_t), \cL) \geq \gamma(\pi_t, \cL)$. 
\end{proof}

With this in mind, when we take $\cE$ to any unitary operation, $\|\cE^{-1}\|_{\diamond} = 1$ and the upper and the lower bounds become equal so that $\gamma(\cE(\pi_t), \cL) = \gamma(\pi_t, \cL)$. On the other hand, this Lemma shows that if the program states are inevitably subjected to the influence of a invertible quantum noise, the cost of the entire process will increase.
\begin{lemma}
    Let $\cL_1, \cL_2$ be two commute  Lindbladians. Suppose $\{\pi_t\}_t$ and $\{\sigma_t\}_t$ are two sets of (exact) veritable program states for $\cL_1$ and $\cL_2$, respectively. Then, there exists program states $\{\omega_t\}_t$ such that,
    \begin{equation*}
        \gamma(\omega_t, \cL_1 + \cL_2) \leq \gamma(\pi_t, \cL_1) + \gamma(\sigma_t, \cL_2).
    \end{equation*}
\end{lemma}
\begin{proof}
    Let $\cL_1:\cB(\cH_1) \rightarrow \cB(\cH_2)$ and $\cL_2: \cB(\cH_2)\rightarrow\cB(\cH_3)$ with $\cH_1\simeq \cH_2 \simeq \cH_3$. Let $\cP$ be the optimal programming channel for $\cL_1$ and $\cQ$ be the map for $\cL_2$. We can define $\omega_t:= \pi_t \ox \sigma_t$ acting on $\cH_P = \cH_{P_1} \ox \cH_{P_2}$, and a programming channel $\cW:=\cQ \circ(\cP \ox \cI_{P_2})$ acting on $\cD(\cH_S\ox \cH_{P})$. Notice that,
    \begin{equation*}
        \cW(\rho, \omega_t) = \cQ \circ(\cP \ox \cI_{P_2})(\rho, \pi_t \ox \sigma_t) = \cQ(\cP(\rho, \pi_t), \sigma_t) = e^{\cL_2 t} \circ e^{\cL_1 t}(\rho) = e^{(\cL_1 +\cL_2)t}(\rho).
    \end{equation*}
    for any $t\geq 0$. Therefore, $\{\cW, \omega_t\}_t$ forms a feasible protocol. Now, due to the  multiplicativity under map composition~\cite{Jiang2021physical,Regula2021operational}, we have,
    \begin{equation*}
        2^{\gamma(\omega_t, \cL_1 + \cL_2)} \leq \|\cQ \circ(\cP \ox \cI_{P_2})\|_{\diamond} \leq \|\cP\|_{\diamond}\| \cQ\|_{\diamond} = 2^{\gamma(\pi_t, \cL_1)} \cdot 2^{\gamma(\sigma_t, \cL_2)},
    \end{equation*}
    as required.
\end{proof}

\begin{lemma}\label{lem:tensor_subadditivity}
    Let $\cL_1$ and $\cL_2$ be two Lindbladians acting on $\cD(\cH_1)$ and $\cD(\cH_2)$, respectively. Suppose $\{\pi_t\}_t$ and $\{\sigma_t\}_t$ are two sets of (exact) veritable program states for $\cL_1$ and $\cL_2$, respectively. Then, there exists program states $\{\omega_t\}_t$ such that,
    \begin{equation*}
        \gamma(\omega_t, \cL_1 \ox \cI_2 + \cI_1 \ox \cL_2) \leq \gamma(\pi_t, \cL_1) + \gamma(\sigma_t, \cL_2).
    \end{equation*}
\end{lemma}
\begin{proof}
    Let $\cP$ be the optimal programming channel for $\cL_1$ and $\cQ$ be the map for $\cL_2$. Denote the total principal system as $\cH_{S}=\cH_1 \ox \cH_2$. We can define $\omega_t:= \pi_t \ox \sigma_t$ acting on $\cH_P = \cH_{P_1} \ox \cH_{P_2}$, and a programming channel $\cW:=(\cP\ox\cQ)\circ \Op{SWAP}_{\cH_2\leftrightarrow \cH_{P_1}}$ acting on $\cD(\cH_S\ox \cH_P)$. Notice that,
    \begin{equation*}
        \cW(\rho, \omega_t) = (\cP\ox\cQ)(\rho, \pi_t \ox \sigma_t) = (e^{\cL_1 t} \ox e^{\cL_2 t})(\rho)
    \end{equation*}
    for any $t\geq 0$. Therefore, $\{\cW, \omega_t\}_t$ forms a feasible protocol. Now, due to the invariance under unitary operation and multiplicativity under `$\ox$'~\cite{Jiang2021physical,Regula2021operational}, we have,
    \begin{equation*}
        2^{\gamma(\omega_t, \cL_1 \ox \cI_2 + \cI_1 \ox \cL_2)} \leq \|\Op{SWAP}_{\cH_2\leftrightarrow \cH_{P_1}}\|_{\diamond}\|\cP\ox \cQ\|_{\diamond} = \|\cP\|_{\diamond}\| \cQ\|_{\diamond} = 2^{\gamma(\pi_t, \cL_1)} \cdot 2^{\gamma(\sigma_t, \cL_2)}.
    \end{equation*}
\end{proof}

Analogous to the studies of physical implementability of general linear maps, a interesting question can be issued from Lemma~\ref{lem:tensor_subadditivity}, that is whether the backward inequality holds, or from a couterpart, whether there exists entangled state $\omega_t$ showcasing the advantage of quantumness by proving $$\gamma(\omega_t, \cL_1 \ox \cI_2 + \cI_1 \ox \cL_2)) < \gamma(\pi_t, \cL_1) + \gamma(\sigma_t, \cL_2).$$
However, this is challenging as there is no efficient method to optimize the veritable $\pi_t$ within the entire state space acting on $\cH_P$. We leave this for a futher investigation.

\begin{proposition}
    Let $\cL_1$ and $\cL_2$ be two $\hptp$-programmable  Lindbladians acting on $\cD(\cH_d)$ with $\pi_t, \sigma_t$ the program states, respectively. Then for any $\epsilon > 0$, there exists $n\in \NN_+$ and program state $\omega_t$ such that,
    \begin{equation*}
        \frac{1}{n}\gamma_{\epsilon}(\omega_t, \cL_1 + \cL_2) \leq \gamma(\pi_t, \cL_1) + \gamma(\sigma_t, \cL_2)
    \end{equation*}
\end{proposition}
\begin{proof}
    Let $\cP_t, \cQ_t$ be defined as
    \begin{equation*}
        \cP_t := \cP(\cdot\ox \pi_t)\; \cQ_t:= \cQ(\cdot \ox \sigma_t),
    \end{equation*}
    where $\cP, \cQ$ are the optimal programming channels of $\cL_1, \cL_2$, respectively. Define $\cW_t:= \cP_t \circ \cQ_t$ so that,
    \begin{equation*}
        \cW_t^n = (\cP_t \circ \cQ_t)^n = \cP(\cQ(\cdots \cP(\cQ(\cdot \ox \sigma_t)\ox \pi_t))) = (e^{\cL_1 t} \circ e^{\cL_2 t})^n .
    \end{equation*}
    $\forall t\geq 0$. Then, for any $\epsilon > 0$, by Trotter product formula, there exists $n\in \NN_+$, such that,
    \begin{equation*}
        \frac{1}{2}\|\cW^n_{t/n} - e^{(\cL_1 + \cL_2) t}\|_{\diamond} \leq \epsilon.
    \end{equation*}
    Define $\omega_t:= \bigotimes_{j=1}^n \pi_{t/n} \ox \sigma_{t/n}$. We then have,
    \begin{equation*}
        2^{\gamma_{\epsilon}(\omega_t, \cL_1 + \cL_2)} \leq \|((\cP\ox\cI) \circ \cQ)^n\|_{\diamond} \leq \| \cP\|^n_{\diamond} \| \cQ\|^n_{\diamond} = 2^{n\gamma(\pi_t, \cL_1)}2^{n\gamma(\sigma_t, \cL_2)}.
    \end{equation*}
    Taking the logarithm of both sides to derive,
    \begin{equation*}
        \frac{1}{n}\gamma_{\epsilon}(\omega_t, \cL_1 + \cL_2) \leq \gamma(\pi_t, \cL_1) + \gamma(\sigma_t, \cL_2)
    \end{equation*}
    as required.
\end{proof}

The Trotter decomposition does provide a convenient strategy to build up the approximate programming protocol from each exact protocol of $\cL_1$ and $\cL_2$. This resembles the framework designed in~\cite{Patel2023waveI,Patel2023waveII}. A interesting observation from the proposition is that when $\cL_1$, $\cL_2$ are CPTP-programmable so that $\gamma(\pi_t, \cL_1) = \gamma(\sigma_t, \cL_2) = 0$. Then, for any $\epsilon > 0$, we have $\gamma_{\epsilon}(\omega_t, \cL_1 +\cL_2) < 0$. It is natural to ask whether $\gamma(\omega_t, \cL_1 +\cL_2) = 0$. In that case, the summation of any CPTP-programmble Lindbladians will again, generate a CPTP-programmable Lindbladian. The (semi-)continuity of the programming cost in terms of the error $\epsilon$ can leads to the target, despite, the cost is generally not continuous at $\epsilon = 0$.

\subsection{Programming dynamics using Choi states}\label{appendix:choi_state_resource}
Given a Lindbladian $\cL$ acting on $\cD(\cH_d)$, a natural choice of analytic program states is to use the Choi states $J(e^{t \cL}) / d$ at time $t$. To see this, since $e^{t \cL}$ is defined analytically. Through the Choi isomorphism, the trajectory $J(e^{t \cL})$ within the space of Choi operators of quantum channels is analytically defined. Using Choi states as the resource states is closely related to the \textit{port-based} quantum channel simulation strategy in the quantum information theory~\cite{Kaur2017amortized,Bennett1996mixed,Horodecki1999general,Chiribella2009realization}. For convenient, we omit writing the $\pi_t$ dependence in the programming cost if $\pi_t = J(e^{t \cL}) / d$. One must be clear that $\gamma(\cL) > 0$ does not guarantee that $\cL$ is not CPTP-programmable even through the Choi states contain all the information of the semigroup element at time $t$. The optimality was only proven for several unitary scenarios~\cite{Vidal2002storing,Sedlak2019optimal,Bisio2010optimal}. Recalling the master equation with any initial state $\rho_0$, 
\begin{equation}\label{eq:de_master_eq_linear_maps}
    \frac{d}{dt} \cA_t(\rho_0) = \frac{d\cA_t}{dt}(\rho_0) = \cL(\rho(t)) = \cL\circ \cA_t(\rho_0) \Rightarrow \frac{d\cA_t}{dt} = \cL\circ \cA_t.
\end{equation}
We then derive a ordinary differential equation defined on the space of linear maps with a fixed initial condition $\cA_0 = \cI$. If the programming channel, exists as an HPTP channel $\cP$, such that $\cP(\pi_t \ox \rho_0) = \cA_t(\rho_0)$. By the linearity of $\cP$,
\begin{equation*}
\begin{aligned}
    \frac{d}{dt} \cP(\pi_t \ox \rho_0) &= \cP(\frac{d\pi_t}{dt} \ox \rho_0) = \frac{d}{dt}\cA_t(\rho_0) = \frac{d}{dt}e^{t \cL}(\rho_0) = \cL\circ \cA_t(\rho_0).
\end{aligned}
\end{equation*}
If we set $\pi_t$ to be the Choi state of $\cA_t$. The derivative can be computed,
\begin{equation*}
    \frac{d\pi_t}{dt} = \frac{1}{d}\sum_{ij} \ketbra{i}{j} \ox \frac{d}{dt}\cA_t(\ketbra{i}{j}) = \frac{1}{d}\sum_{ij} \ketbra{i}{j} \ox \cL\circ \cA_t(\ketbra{i}{j}) = \frac{J(\cL\circ \cA_t)}{d}.
\end{equation*}
We then have an equivalent condition for the programmability of $\cL$ as,
\begin{equation*}
    J({\cP}) \star J({\cL}) \star J(\cA_t) = dJ({\cL}) \star J(\cA_t) \quad \forall t\geq 0.
\end{equation*}
\begin{lemma}\label{lem:physical_trans}
    Let $A$ be some $d^2 \times d^2$ linear operator and $\cP$ be a CPTP map from $\cB(\cH_{d^3})$ to $\cB(\cH_{d})$. Then, $(\cI \ox \cP)(\Phi \ox A) = A$ iff.
    \begin{equation*}
        \cP(\ketbra{i}{j} \ox A) = dA_{ij}
    \end{equation*}
    for $\{\ket{j}\}_j$ a set of orthonormal basis of $\cH_{d}$ and $A_{ij}$ is $(i,j)$th $d\times d$ block matrix of $A$ i.e., $A = \sum_{ij} \ketbra{i}{j} \ox A_{ij}$.
\end{lemma}
\begin{proof}
    Let $\{F_{l}\}_l$ be the Kraus representation of $\cP$, and $\hat{F}_l= (I_d \ox F_l)(\ket{\Phi} \ox I_{d^2})$ be defined as before. Then, we have,
    \begin{equation*}
    \begin{aligned}
        (\cI \ox \cP)(\Phi \ox A) &= \sum_l \hat{F}_l A \hat{F}_l^{\dagger} = \sum_l (I_d \ox F_l)(\ket{\Phi} \ox I_{d^2})(I_{d^2}\ox A)(\bra{\Phi} \ox I_{d^2})(I_d \ox F^{\dagger}_l )\\
        &= \frac{1}{d}\sum_{lij} \ketbra{i}{j} \ox F_l (\ketbra{i}{j} \ox A) F_l^{\dagger} \\
        &= \frac{1}{d}\sum_{ij} \ketbra{i}{j} \ox \sum_l  F_l (\ketbra{i}{j} \ox A) F_l^{\dagger} \\ &= \frac{1}{d}\sum_{ij} \ketbra{i}{j} \ox \cP(\ketbra{i}{j} \ox A ).
    \end{aligned}
    \end{equation*}
    Since $(\cI \ox \cP)(\Phi \ox A) = A$. We then, have $\cP(\ketbra{i}{j}\ox A) = dA_{ij}$ as required.
\end{proof}

From the Lemma, the action of a such $\cP$ is concluded as a scaling of each $d$-by-$d$ block matrix regarding a basis of $\cH_d$. The above calculations does not provide a general theory of programmability of $\cL$. However, it does provide some insights for the $\gamma(\cL)$ for a specific class of Lindbladian. 
\begin{definition}[Steady states]
    Let $\cL$ be a Lindbladian, we denote $\Op{St}[\cL]$ as the set of steady states of $\cL$, i.e.,
    \begin{equation*}
        \Op{St}[\cL]:=\{\rho\in \cD(\cH_d): \; \cL(\rho) = 0\}
    \end{equation*}
\end{definition}
Clearly, the set $\Op{St}[\cL]$ is a convex subset of $\cD(\cH_d)$ since for any $\rho, \sigma \in \Op{St}[\cL]$, the convex combination $\rho' = p\rho + q \sigma$ for $p,q\geq 0$ and $p+q = 1$,
\begin{equation*}
    \cL(p\rho + q \sigma) = p\cL(\rho) + q \cL(\sigma) = 0.
\end{equation*}
In the traditional studies of Lindbladians. The space of steady states and its orthogonal spaces fundamentally dominate the asymptotic behavior of the Lindbladian dynamics. Particularly, a Lindbladian has a unique steady state governs a \textit{ergodic} system which contains no ``hidden'' conserved quantities or symmetries in the dissipative part of the dynamics. The environment thoroughly ``scrambles'' the system, erasing all memory of its initial state.
\begin{theorem}\label{prop:replace_channel_at_infty}
    Let $\cL$ be a Lindbladian such that $\Op{St}[\cL] = \{\sigma\}$. If there exists an eigenvalue $\lambda < 1/d^2$. Then, $\gamma(\cL) > 0$.
\end{theorem}
\begin{proof}
    We prove the proposition by contradiction. Let $\lambda<1/d^2$ be the some eigenvalue of the replaced state $\sigma$. Then, the set,
    \begin{equation*}
        \cS(\sigma):=\{\tau\in\cD(\cH_d):\; F(\tau, \sigma) < 1/d^2\}
    \end{equation*}
    is non-empty. Taking $\tau\in \cS(\sigma)$ as a fixed state having an eigendecomposition $\tau = \sum_j p_j\ketbra{\psi_j}{\psi_j}$. Then we can construct an orthonormal basis of $\cH_d$ denoted as $\cB:=\{\ket{\psi_j^*}\}_{j=0}^{d-1}$ where $\ket{\psi_j^*}$ is the complex conjugate of $\ket{\psi_j}$. The Choi states read $J(\cA_0) = \Phi, J(\cA_{\infty}) = \frac{I_d}{d} \ox \sigma$. We can expand $\Phi$ using the basis $\cB$ such that,
    \begin{equation*}
        \Phi = \sum_{ij} \ketbra{\psi_i^*}{\psi_j^*} \ox \Phi_{ij},
    \end{equation*}
    where $\Phi_{ij} = \tr_1[(\ketbra{\psi_i^*}{\psi_j^*} \ox I_d)\Phi] = \ketbra{\psi_i}{\psi_j}$. Now, suppose a $\cptp$ programming channel $\cP$ exists. From Lemma~\ref{lem:physical_trans}, using linearity of $\cP$, we have,
    \begin{equation}
    \begin{aligned}
        \cP(\sigma^* \ox J(\cA_0)) &=  \sum_j p_j\cP(\ketbra{\psi_j^*}{\psi_j^*} \ox J(\cA_0)) = \sum_j p_j\ketbra{\psi_j}{\psi_j} = \tau; \\
        \cP(\sigma^* \ox J(\cA_{\infty})) &= \sum_j p_j\cP(\ketbra{\psi_j^*}{\psi_j^*} \ox J(\cA_{\infty})) = (\sum_j p_j)\sigma = \sigma.
    \end{aligned}    
    \end{equation}
    Notice that the fidelity,
    \begin{equation*}
        F(\sigma^* \otimes J(\cA_0), \sigma^* \otimes J(\cA_{\infty})) = F(\sigma^*, \sigma^*)\cdot F(J(\cA_0), J(\cA_{\infty})) = \frac{1}{d^2},
    \end{equation*}
    where the last equality holds due to the fact that,
    \begin{equation*}
        \bra{\Phi}\frac{I}{d}\ox \sigma\ket{\Phi} = \frac{1}{d^2}\sum_{ij} \delta_{ij} \bra{j}\sigma \ket{i} = \frac{1}{d^2}.
    \end{equation*}
    However, notice that $F(\tau, \sigma) < 1/d^2$, which violate the data processing inequality. Therefore, the existence of $\cP$ causes a  contradiction.
\end{proof}

This theorem provides a simple criterion for determining whether or not the ergodic system can be physically programmed via port-based strategy. A direct observation is that when $\cA_{\infty}:=\lim_{t\rightarrow \infty} \cA_t$ is a replacement channel with a pure fixed state. The theorem automatically applies as the null space of a pure state is not empty and $\gamma(\cL) > 0$. On the other hand, if the fixed point state is nearly the maximally mixed state, the assumption of the theorem does not meet and we can not conclude the programmability of $\cL$. A famous example is the isotropic depolarzing system with $\Op{St}[\cL] = \{I_d / d\}$, which has been proven to be teleportation-simulatable. From our previous results, we can even restrict the $d_P = 2$ for any dimension $d$ to physically program the dynamics.

\end{document}